\documentclass[10pt,twocolumn,twoside]{IEEEtran}
\usepackage{graphics} % for pdf, bitmapped graphics files
\usepackage{balance}
\usepackage{setspace}
\usepackage{pdfsync}
\usepackage{epstopdf}
\usepackage{subfigure}
\usepackage{amsmath, amsfonts, amssymb,amscd}
\usepackage[ruled,vlined]{algorithm2e} % For writing the algorithm

\newtheorem{theorem}{Theorem}

\newtheorem{remark}{Remark}
\newtheorem{definition}{Definition}
\newtheorem{example}{Example}
\newtheorem{proof}{Proof}

\usepackage{epsfig} % for postscript graphics files
\usepackage{psfrag}
\usepackage{amsmath} % assumes amsmath package installed

\usepackage{color}

\newcommand{\qed}{\hfill \ensuremath{\Box}}

\usepackage{mathtools}
 \usepackage{setspace}
 \usepackage{epsfig}
 \usepackage[ruled,vlined]{algorithm2e} % For writing the algorithm

\usepackage{cite}
\def\R{\mathbb{R}}

\usepackage{color}

\allowdisplaybreaks
\begin{document}

\title{{Trustworthy Distributed Average Consensus based on Locally Assessed Trust Evaluations}}
\author{Christoforos N. Hadjicostis,~\IEEEmembership{Fellow,~IEEE} and Alejandro D. Dom\'{i}nguez-Garc\'{i}a,~\IEEEmembership{Fellow,~IEEE}
\thanks{C. N. Hadjicostis is with the ECE Department at the University of Cyprus, Nicosia, Cyprus, and also with the ECE Department at  the University of Illinois at Urbana-Champaign, Urbana, IL 61801, USA. E-mail:  hadjicostis.christoforos@UCY.AC.CY.}
\thanks{A. D. Dom\'{i}nguez-Garc\'{i}a is with the ECE Department at the University of Illinois at Urbana-Champaign, Urbana, IL 61801, USA. E-mail:  aledan@ILLINOIS.EDU.}
}
\maketitle
\begin{abstract}
This paper proposes a distributed algorithm for average consensus in a multi-agent system under a fixed bidirectional communication topology, in the presence of malicious agents (nodes) that may try to influence the average consensus outcome by manipulating their updates. The proposed algorithm converges asymptotically to the average of the initial values of the non-malicious nodes, which we refer to as the {\em trustworthy average}, as long as the underlying topology that describes the information exchange among the non-malicious nodes is connected. We first present a distributed iterative algorithm that assumes that each node receives (at each iteration or periodically) side information about the trustworthiness of the other nodes, and it uses such trust assessments to determine whether or not to incorporate messages received from its neighbors, as well as to make proper adjustments in its calculation depending on whether a previously trustworthy neighbor becomes untrustworthy or vice-versa.
We show that, as long as the trust assessments for each non-malicious node eventually reflect correctly the status (malicious or non-malicious) of its neighboring nodes, the algorithm guarantees asymptotic convergence to the trustworthy average. We subsequently discuss how the proposed algorithm can be enhanced with functionality that enables each node to obtain trust assessments about its neighbors by utilizing information that it receives from its two-hop neighbors at infrequent, perhaps randomly chosen, time instants.

\end{abstract}

\section{Introduction and Motivation}

Average consensus and, more generally, consensus and distributed function calculation have received attention by many communities, including the control community (which has considered applications in multi-agent systems, formation control, and sensor networks), the communication community, and the computer science community \cite{1984:Tsitsiklis,1996:Lynch,2003:Koetter,2004:Murray,2005:Hromkovic,2008:Cortes}. In particular, average consensus has been studied extensively, primarily in settings where convergence is asymptotic and each node processes and transmits real-valued states with infinite precision \cite{2004:Murray, 2003:jadbabaie_coordination, cao2008reaching, 2009:Olshevsky}; however, issues of finite time completion \cite{sundaram2007finite, wang2010finite, yuan2013decentralised, hendrickx2014finite}, quantized transmissions \cite{kashyap2007quantized, aysal2008distributed, 2009:Nedic, lavaei2011quantized, cai2011quantized}, and event-triggered operation \cite{seyboth2013event, nowzari2019event, rikos2020event} have also been considered. Reference \cite{2018:BOOK} discusses several applications of distributed average consensus.

This paper addresses asymptotic average consensus in the presence of malicious nodes, which try to influence the outcome of the distributed computation by arbitrarily manipulating their initial values and/or their updates, possibly in a colluding manner. This is a topic that has recently received some attention as described later in this section. In this paper, we propose and analyze a novel distributed algorithm, which enables the non-malicious nodes of a distributed system to distributively identify and isolate malicious nodes, and to eventually calculate the \textit{exact} average of their initial values, despite the actions of malicious nodes. The communication topology is assumed to be bidirectional and we require that (i) the induced topology when restricting attention to non-malicious nodes is connected, and (ii) non-malicious nodes have access, perhaps periodically, to certain information provided by their two-hop neighbors (i.e., the neighbors of their neighbors).
 
The proposed scheme essentially takes a rather standard distributed algorithm for average consensus in bidirectional communication topologies (that relies on linear updates with weights that form a doubly stochastic matrix---see, e.g., \cite{2018:BOOK}) and enhances it in two ways. First, by having each node maintain one additional {\em running-sum} variable for each of its neighbors, we devise a scheme that allows each node to virtually remove from (or add to) the distributed computation neighboring nodes that become untrustworthy (or trustworthy). This is done in a way that ensures that, if all non-malicious nodes eventually learn the trustworthiness of all of their neighbors, then they will asymptotically converge to the {\em trustworthy average} (i.e., the average of the initial values of the non-malicious nodes), despite the actions (e.g., erroneous computational updates) and ignoring the initial values of the malicious nodes. Second, by establishing an invariant that holds during the execution of the proposed algorithm, we devise a scheme that allows each node to determine whether its neighbors are trustworthy or not. The checking is done in a distributed manner assuming each node has infrequent access to information sent by its two-hop neighbors (i.e., by the neighbors of its neighbors).

The proposed algorithm (see Algorithm~\ref{algorithm_0}) removes/adds untrustworthy/trustworthy nodes utilizing trust evaluations that become available at each iteration. It possesses some of the features of the variation of running-sum ratio consensus algorithm we proposed in \cite{CDCpaper}, which can be used by non-malicious nodes in fixed, possibly directed communication topologies, to asymptotically converge to the {\em exact} value of the trustworthy average, as long as the induced graph obtained by focusing on non-malicious nodes is strongly connected, and the non-malicious nodes eventually learn which nodes among their in-neighbors and out-neighbors are malicious. However, the work in \cite{CDCpaper} did not identify an invariant and did not discuss how the non-malicious nodes can exploit it to obtain the needed trust assessments about each neighbor.

It is worth pointing out that the idea of adding and removing nodes from a distributed average consensus computation appears in works that rely on the running-sum ratio consensus algorithm (e.g., \cite{nitinTAC2016,eyal2014limosense}) or in works that deal with dynamic average consensus (e.g., \cite{zhu2010discrete, montijano2014robust}). Unlike the work in this paper, however, the aforementioned works deal with nodes that willingly remove themselves (or collaborate with their neighboring nodes in order to remove themselves) from the computation of the average. Instead, the solution proposed in this paper (as also the solution in \cite{CDCpaper}), ensures that all influence that a malicious node had on the distributed computation (including past influence) is nullified by its non-malicious neighbors. 

Related ideas about utilizing trust assessments towards trustworthy average consensus appear in \cite{yemini2021characterizing}, which considers a distributed system where stochastic values of trust between the nodes are available, and nodes use these trust values to reach agreement to a common limit value. The authors of \cite{yemini2021characterizing} show that, under certain conditions on the trust values, the deviation of the common limit value that is reached from the true consensus value is bounded in a manner that can be characterized. Moreover, correct classification of malicious and non-malicious nodes can be attained in finite time almost surely. Unlike the setting in \cite{yemini2021characterizing}, the setting in this paper is deterministic and the solution we propose is based on a completely different algorithm that guarantees convergence to the exact average of the trustworthy nodes.

The proposed scheme for obtaining trust assessments can be embedded in the proposed distributed average consensus algorithm and exploits information from two-hop neighbors. It does not require a mechanism to inform all nodes about which specific node has been identified as malicious (because trust assessments are only needed for neighboring nodes, and each node has a direct way of obtaining the trust assessments it needs). Unlike \cite{yuan2019resilient, yuan2021resilient, yuan2021secure} (which also exploit information from two-hop neighbors), the proposed scheme for obtaining trust assessments can be performed at each node by requiring information from two-hop neighbors, infrequently or even at randomly selected points of time, which significantly reduces the communication overhead; this is achieved by exploiting an invariant that holds during the execution of the proposed algorithm, as shown in this paper. 

Related ideas about stochastic side information that can serve as trust assessments also appear in \cite{agkun2022}, but the focus in that paper is about how the different nodes can learn whether to trust the other agents. The assumption is that nodes can directly evaluate whether to trust their in-neighbors and the focus is on how to propagate trust assessments from other nodes to their out-neighbors and eventually, via a consensus mechanism, to the whole network.

In addition to \cite{CDCpaper, yuan2019resilient,yuan2021resilient, yuan2021secure, agkun2022} discussed above, there is some related work on distributed average consensus algorithms that aim at removing/limiting the effect of malicious nodes on the computation. For example, the authors of \cite{sundaram2011distributed} exploit the connectivity of the underlying topology in order to detect and isolate the effects of malicious nodes when performing distributed function calculation. More specifically, if the underlying graph is $(2f+1)$-connected (i.e., any two non-neighboring nodes have at least $2f+1$ node-disjoint paths that connect them), then one can systematically exploit information that arrives from these disjoint paths in order to withstand up to $f$ malicious nodes. Alternatively, one could use the approach proposed in \cite{leblanc2012consensus} to limit the effect of malicious nodes on the computation; however, such approach would only guarantee reaching consensus to a value between the minimum and maximum value of non-malicious nodes (not necessarily the average).

The remainder of the paper is organized as follows. In Section~\ref{SECpreliminaries}, we introduce some preliminary concepts from graph theory and distributed averaging over bidirectional communication topologies. In Section~\ref{SECtrustalgo}, we formulate the trustworthy distributed average consensus problem and outline an algorithm, including its pseudocode description, to solve it. In this section, we also establish an invariant that holds during the execution of the algorithm and which is used to prove the algorithm's correctness. In Section~\ref{SECexamples}, we provide simulation studies to illustrate the operation of the algorithm, under different scenarios for the convergence of the trust assessments and the behavior of the malicious nodes. Section~\ref{SECtrust} describes the scheme for distributively obtaining trust assessments while executing the proposed distributed algorithm. Finally, Section~\ref{SECconclusions} concludes with some directions for future research.

\section{Mathematical Background and Notation}
\label{SECpreliminaries}

In this section we provide some needed background on graph theory and review some existing distributed protocols for reaching average consensus in multi-agent systems,

\subsection{Graph-Theoretic Notions and Communication Topology}

A directed graph (digraph) of order $N$ ($N \geq 2$), is defined as $\mathcal{G} = (\mathcal{V}, \mathcal{E})$, where $\mathcal{V} =  \{v_1, v_2, \dots, v_N\}$ is the set of vertices (nodes) and $\mathcal{E} \subseteq \mathcal{V} \times \mathcal{V} - \{ (v_j,v_j)$ $|$ $v_j \in \mathcal{V} \}$ is the set of links (edges). A directed edge from node $v_i$ to node $v_j$ is denoted by $(v_j, v_i) \in \mathcal{E}$, and indicates that node $v_i$ can send information to node $v_j$. A digraph is called \textit{strongly connected} if for each pair of nodes $v_j, v_i \in \mathcal{V}$, $v_j \neq v_i$, there exists a directed \textit{path} from $v_i$ to $v_j$, i.e., we can find a sequence of nodes $v_i =: v_{l_0},v_{l_1}, \dots, v_{l_t} := v_j$ such that $(v_{l_{\tau+1}},v_{l_{\tau}}) \in \mathcal{E}$ for $ \tau = 0, 1, \dots , t-1$. All nodes that can send information to node $v_j$ directly are said to be its in-neighbors and belong to the set $\mathcal{N}_j^- = \{v_i \in \mathcal{V} \; | \; (v_j, v_i) \in \mathcal{E} \}$, the cardinality of which is referred to as the \textit{in-degree} of $v_j$ and is denoted by $D_j^-$. The nodes that can receive information from node $v_j$ are said to be its out-neighbors and belong to the set $\mathcal{N}_j^+ = \{v_l \in \mathcal{V} \; | \; (v_l, v_j) \in \mathcal{E} \}$, the cardinality of which is referred to as the \textit{out-degree} of $v_j$ and is denoted by $D_j^+$.

In this paper, we consider multi-agent systems in which the exchange of information between a pair of nodes, if allowed, is bidirectional. Then, the communication topology of the multi-agent system can be described by a bidirectional communication graph, which we define as follows.

\begin{definition}
A digraph $\mathcal{G}=(\mathcal{V}, \mathcal{E})$ is called a bidirectional communication graph if $(v_j, v_i) \in \mathcal{E}$ implies that $(v_i, v_j) \in \mathcal{E}$.
\end{definition}

Under a bidirectional communication graph $\mathcal{G}$, all nodes that can send/receive information to/from node $v_j$ directly are said to be its neighbors and belong to the set $\mathcal{N}_j = \{v_i \in \mathcal{V} \; | \; (v_j, v_i) \in \mathcal{E} \} = \{v_l \in \mathcal{V} \; | \; (v_l, v_j) \in \mathcal{E} \}$, which satisfies $\mathcal{N}_j = \mathcal{N}_j^+ = \mathcal{N}_j^-$. The cardinality of $\mathcal{N}_j$ is referred to as the \textit{degree} of $v_j$ and is denoted by $D_j = | \mathcal{N}_j |$. A bidirectional communication graph is said to be {\em connected} if, for each pair of nodes $v_j, v_i \in \mathcal{V}$, $v_j \neq v_i$, there exists a \textit{path} from $v_i$ to $v_j$ i.e., we can find a sequence of nodes $v_i =: v_{l_0},v_{l_1}, \dots, v_{l_t} := v_j$ such that $(v_{l_{\tau+1}},v_{l_{\tau}}) \in \mathcal{E}$ (thus, also $(v_{l_{\tau}},v_{l_{\tau+1}}) \in \mathcal{E}$) for $ \tau = 0, 1, \dots , t-1$.

In this paper, we assume a broadcast model under a fixed bidirectional communication graph $\mathcal{G}$. Specifically, when node $v_j$ broadcasts information, its transmissions are received at all of its neighbors in the set $\mathcal{N}_j$; similarly, node $v_j$ receives all transmissions sent by each of its neighbors in the set $\mathcal{N}_j$. Note, however, that node $v_j$ may choose to ignore a transmission from a certain neighbor $v_i$ (e.g., because it considers $v_i$ to be untrustworthy); such actions by transmitting/receiving nodes may result in a virtual communication topology that is not necessarily a bidirectional communication graph.

\noindent
{\bf Assumption 0.} Each transmission by node $v_j \in \mathcal{V}$ is received by {\em all} neighbors of node $v_j$ (i.e., all nodes in the set $\mathcal{N}_j$). Furthermore, we assume that each transmission is associated with a unique node ID that allows receiving nodes to identify the sending node.

\subsection{Average Consensus via Linear Iterations}
\label{SUBSECrunning}

Consider a distributed system, captured by a bidirectional communication graph $\mathcal{G} = (\mathcal{V}, \mathcal{E})$, in which each node $v_j \in \mathcal{V}$ has a value $x_j$. Average consensus aims to have all the nodes calculate $\overline{X} = \frac{\sum_{\ell=1}^N x_\ell}{N}$ in a distributed manner. This can be achieved via a linear iteration where each node $v_j$ maintains a scalar state variable $x_j[k]$, which it updates based on the values received from its neighbors. Specifically, each node $v_j$ uses a linear time-invariant update of the form 
\begin{equation}
x_j[k+1]=w_{jj} x_j[k]+ \sum_{v_i \in \mathcal{N}_j} w_{ji} x_i[k] \; , \label{EQgeneral}
\end{equation}
where $x_j[0]=x_j$, $v_j \in \mathcal{V}$, and the $w_{ji}$'s are constant weights. If we let $x[k]=[x_1[k], x_2[k], \dots, x_N[k]]^{\mathrm{T}}$, then the iteration in \eqref{EQgeneral} can be written compactly in matrix form as
\begin{equation}
x[k+1] =W x[k] \; , \;\; x[0] =[ x_1, x_2, \dots, x_N]^{\mathrm{T}}, \label{EQgeneral_matrix}
\end{equation}
where $W = [ w_{ji}] \in \R^{N \times N}$ is referred to as the weight matrix, with the entry $w_{ji}$ at its $j$th row and $i$th column such that $w_{ji} =0$ if $v_i \notin \mathcal{N}_j \cup \{ v_j \}$. The nodes are said to reach asymptotic average consensus if
\begin{equation}
\lim_{k \rightarrow \infty} x_j[k] = \overline{X} \; , \quad \forall v_j \in \mathcal{V} \; . \label{EQasymptoticconsensus}
\end{equation}

The necessary and sufficient conditions for the iteration in \eqref{EQgeneral_matrix} to asymptotically reach average consensus are \cite{2004:XiaoBoyd, 2004:Murray}: (i) $W$ has a simple eigenvalue at $1$, with left eigenvector $1^T_N$ and right eigenvector $1_N$ (where $1_N$ denotes the $N$-dimensional all-ones column vector), and (ii) all other eigenvalues of $W$ have magnitude strictly less than $1$. If one focuses on nonnegative weights, these conditions are equivalent to $W$ being a primitive doubly stochastic matrix. In the case of a bidirectional communication graph, there are very simple ways for the nodes to choose the weights, in a distributed manner, so that $W$ forms a doubly stochastic matrix (see, e.g., \cite{2018:BOOK}). For example, assuming the nodes know the total number of nodes $N$ or an upper bound $N' \geq N$, each node $v_j$ can set the weights on all of its incoming links to be $w_{ji} = \frac{1}{N'}$ for all $v_i \in \mathcal{N}_j$ and $w_{jj} = 1-\frac{D_j}{N'}$ (zero otherwise). It is easy to verify that $W$ will be a (symmetric) doubly stochastic matrix. Furthermore, $W$ will be primitive as long as $\mathcal{G}$ is connected.

The linear iteration in \eqref{EQgeneral} can also be extended to a time-varying topology setting as follows. Consider that at each iteration $k$, the topology is captured by a bidirectional communication graph $\mathcal{G}[k] = (\mathcal{V}, \mathcal{E}[k])$, where the set of nodes remains fixed but the set of edges $\mathcal{E}[k]$ can vary for different values of $k$. We can now consider the iteration
\begin{equation}
x_j[k+1]= w_{jj}[k] x_j[k] + \sum_{v_i \in \mathcal{N}_j[k]} w_{ji}[k] x_i[k] \; , \label{EQgeneraltime}
\end{equation}
with $x_j[0]=x_j$, $v_j \in \mathcal{V}$, and time-varying weights $w_{ji}[k]$, where $\mathcal{N}_j[k]$ is the set of neighbors of node $v_j$ in $\mathcal{G}[k]$. We can easily choose the time-varying weights to form a matrix $W[k]=[w_{ji}[k]]$ that is doubly stochastic and conforms to the topology captured by $\mathcal{G}[k]$, such that $w_{ji}[k] \geq c$ for $v_i \in \mathcal{N}_j[k] \cup \{ v_j \}$, where $c$ is some positive constant. For example, if we let $D_j[k] = | \mathcal{N}_j[k] |$ denote the number of neighbors of node $v_j$ at iteration $k$, we can have each node $v_j$ set the weights on all of its incoming links to be $w_{ji}[k] = \frac{1}{N'}$ for all $v_i \in \mathcal{N}_j[k]$ and $w_{jj}[k] = 1-\frac{D_j[k]}{N'}$ (zero otherwise); this results in a weight matrix $W[k]$ that is symmetric and doubly stochastic, but not necessarily primitive (that will depend on whether or not $\mathcal{G}[k]$ is connected). We can write \eqref{EQgeneraltime} in matrix form as
\begin{equation}
x[k+1] =W[k] x[k] \; , \;\; x[0] =[ x_1, x_2, \dots, x_N]^{\mathrm{T}}, \label{EQgeneral_matrix_time}
\end{equation}
and one can show that average consensus is reached under some mild joint connectivity conditions on the graphs $\mathcal{G}[k]$, $k=0, 1, 2, ...$. For instance, it can be shown (see, e.g., \cite{2018:BOOK}) that asymptotic average consensus in \eqref{EQasymptoticconsensus} is reached if we can find a finite $K$ such that each union graph
$$
\begin{array}{l}
\mathcal{G}[\tau K] \cup \mathcal{G}[\tau K+1] \cup \ldots \cup \mathcal{G}[\tau K+K-1] 
\\
\;\;\; := (\mathcal{V}, 
\mathcal{E}[\tau K] \cup \mathcal{E}[\tau K+1] \cup \ldots \cup \mathcal{E}[\tau K+K-1])
\end{array}
$$
for $\tau=0,1,2,\dots$, is connected. 

Note that when implementing the distributed algorithm in \eqref{EQgeneral_matrix} (or, more generally, in \eqref{EQgeneral_matrix_time}), each node $v_j$ simply needs to broadcast its value $x_j[k]$ at iteration $k$; at the same time, node $v_j$ receives the values $\{ x_i[k] \; | \; v_i \in \mathcal{N}_j \}$ (or, more generally, $\{ x_i[k] \; | \; v_i \in \mathcal{N}_j[k] \}$). For notational simplicity, in our development in the remainder of this paper, we make the following assumption:\footnote{Please note that (with some minor adjustments and at the expense of heavier notation) any set of weights that form a doubly stochastic matrix $W$ (or $W[k]$ at each iteration $k$) can be used as long as they conform to the communication topology $\mathcal{G}$ (or $\mathcal{G}[k]$) and $w_{ji}[k] \geq c$ for $v_i \in \mathcal{N}_j[k] \cup \{ v_j \}$, where $c$ is some positive constant.}

\noindent
{\bf Assumption 1.} For each node $v_j \in \mathcal{V}$, the nonzero weights $w_{ji}$ for $v_i \in \mathcal{N}_j$ when implementing iteration \eqref{EQgeneral_matrix} (or $w_{ji}[k]$ for $v_i \in \mathcal{N}_j[k]$ when implementing iteration \eqref{EQgeneral_matrix_time}) are equal to $1/N$ and $w_{jj}$ (or $w_{jj}[k]$) is equal to $1-D_j/N$ (or $1-D_j[k]/N$). 

In our developments later in the paper, we will find it necessary to have the nodes execute a variation of the distributed averaging time-varying iteration of \eqref{EQgeneraltime}, where each node $v_j$ maintains a running-sum variable \cite{nitinTAC2016}, defined as
\begin{equation}
\sigma_j[k+1] := \sum_{t=0}^k x_j[t] \; , \label{EQrunningsum}
\end{equation}
and transmits, at iteration $k$, the running sum $\sigma_j[k+1]$ instead of $x_j[k]$. Then, the iteration in \eqref{EQgeneraltime} (with the weights in Assumption~1) can be executed by each node as follows:
\begin{equation}
x_j[k+1]= \left ( 1-\frac{D_j[k]}{N} \right ) x_j[k] + \frac{1}{N} \sum_{v_i \in \mathcal{N}_j[k]} (\rho_{ji}[k+1] - \rho_{ji}[k]) \label{EQgeneral_matrix_time_node}
\end{equation}
where $\rho_{ji}[k+1] = \sigma_i[k+1]$ for $v_i \in \mathcal{N}_j$. Note that in order to implement this running-sum based version of the iteration, each node $v_j$ maintains one variable for its running sum $\sigma_j$ (which it can easily update as $\sigma_j[k+1] = \sigma_j[k]+x_j[k]$) and also $D_j$ additional variables, namely $\{ \rho_{ji}[k] \; | \; v_i \in \mathcal{N}_j \}$, to remember the previous value of the running sum of each neighbor.

\section{Trustworthy Distributed Averaging}
\label{SECtrustalgo}

\subsection{Problem Formulation}

We are given a bidirectional communication graph $\mathcal{G} = (\mathcal{V}, \mathcal{E})$, which describes the (fixed) topology among a set of nodes in a distributed system. We assume a broadcast model as described in Assumption~0. Each node $v_j \in \mathcal{V}$ has a value $x_j$. A certain subset $\mathcal{V}_T$, $\mathcal{V}_T \subseteq \mathcal{V}$, of the nodes is trustworthy (non-malicious), whereas the remaining nodes in $\mathcal{V}_M = \mathcal{V} \setminus \mathcal{V}_T$ are malicious (untrustworthy). Malicious nodes can collude to behave unpredictably during the execution of the algorithm. The goal of the trustworthy nodes is to compute the {\em average of the trustworthy nodes}, defined by 
\begin{equation}
\label{EQtrustaverage}
\overline{X}_T = \frac{\sum_{v_l \in \mathcal{V}_T} x_l}{|\mathcal{V}_T|} \; ,
\end{equation}
despite any incorrect updates by the malicious nodes. We make the following assumption.

\noindent
{\bf Assumption 2.} The bidirectional communication graph induced from $\mathcal{G}$ by restricting attention to the trustworthy nodes, denoted by $\mathcal{G}_{T} = (\mathcal{V}_T, \mathcal{E}_T)$, where $\mathcal{E}_T = \{ (v_j, v_i) \in \mathcal{E} \; | \; v_j, v_i  \in \mathcal{V}_T \}$, is connected.

\subsection{Trust Assessment Model}

At each iteration $k$, each node $v_j \in \mathcal{V}_T$ has access to an assessment about the trustworthiness of each other node.\footnote{It will become obvious in our development that node $v_j$ only needs information about the trustworthiness of its neighbors, and not necessarily all other nodes; however, for ease of notation, we assume that such information is made available at node $v_j$ for all other nodes.} This assessment could be derived based on measurements of some sort, such as the {\em stochastic values of trust}, $a_{ij} \in (0, 1)$, used in \cite{yemini2021characterizing, gil2017guaranteeing}, which approach $1$ when node $v_i$ should be trusted by node $v_j$, and approach $0$ when node $v_i$ should not be trusted by node $v_j$; or it could be based on checks like the ones in \cite{yuan2019resilient, yuan2021resilient, yuan2021secure}, where, using two-hop communication in bidirectional communication graphs, neighbors of node $v_i$ assess the computations performed by node $v_i$ in order to determine their correctness. 

For now, we abstract away from the specific mechanism of measuring and assessing trust (this issue is addressed in Section~\ref{SECtrust}), and assume that, at each iteration $k$, each node $v_j$ has access to an assessment about the trustworthiness of another node $v_i$. In particular, $t_{ij}[k] \in \{ 0, 1 \}$ is a binary indicator that captures the trustworthiness of node $v_i$ as perceived by node $v_j$ at iteration $k$: $t_{ij}[k]=1$ ($t_{ij}[k]=0$) indicates that node $v_j$ considers node $v_i$ to be trustworthy (malicious) at iteration $k$. We use $\mathcal{T}_j[k] = \{ v_i \; | \; t_{ij}[k]=1 \}$ to denote the set of nodes that are considered trustworthy by node $v_j$ at iteration $k$; we assume that $t_{jj}[k]=1$ and thus $v_j \in \mathcal{T}_j[k]$ for all $k$ (i.e., node $v_j$ always trusts itself). We also use $\mathcal{M}_j[k] = \mathcal{V} \setminus \mathcal{T}_j[k]$ to denote the nodes that are considered malicious by node $v_j$ at iteration~$k$.

Without loss of generality, we assume that initially $\mathcal{T}_j[0] = \mathcal{V}$ for each node $v_j$ (i.e., at the start of the algorithm execution, each node $v_j$ considers all of its neighbors to be trustworthy). We also require that asymptotically $\lim_{k \rightarrow \infty} \mathcal{T}_j[k] = \mathcal{V}_T$, at least for nodes $v_j$ that are trustworthy. Note that convergence of $\mathcal{T}_j[k]$ to $\mathcal{V}_T$ as $k$ goes to infinity does not have to be monotonic as $t_{ij}[k]$ may fluctuate between $1$ and $0$; however, $t_{ij}[k]$ has to eventually settle to $0$ if $v_i$ is malicious or $1$ if $v_i$ is trustworthy. Also, note that at any given $k$, $t_{ij}[k]$ and $t_{ji}[k]$ need not coincide, e.g., node $v_j$ may trust node $v_i$ but not vice-versa. Finally, two different trustworthy nodes, $v_j$ and $v_{l}$ ($v_j, v_l \in \mathcal{V}_T$), may have different trust assessments about node $v_i$ at a given iteration $k$ (i.e., $t_{ij}[k] \neq t_{il}[k]$); however, we require that eventually these assessments would have to be equal and correctly reflect the status of node $v_i$ (i.e., for large $k$, $t_{ij}[k] = t_{il}[k] = 1$ if node $v_i$ is trustworthy and $t_{ij}[k] = t_{il}[k] = 0$ otherwise). This assumption is stated below.

\noindent
{\bf Assumption 3.} The trust assessments $t_{ij}[k]$, $v_i, v_j \in \mathcal{V}$, are such that for each $v_j \in \mathcal{V}_T$, there exists a finite $k_j$ such that 
$$
\mathcal{T}_j[k] := \{ v_i \; | \; t_{ij}[k]=1 \} = \mathcal{V}_T \; , \text{ for } k \geq k_j \; .
$$
Clearly, for $k \geq k_{\max} : = \max_{v_j \in \mathcal{V}_T} \{ k_j \}$, we have $\mathcal{T}_j[k] = \mathcal{V}_T$ for all $v_j \in \mathcal{V}_T$.

\subsection{Trust Assessment-Based Average Consensus}

The trustworthy distributed calculation of the average $\overline{X}_T$ in \eqref{EQtrustaverage} is based on a variation of the linear iteration in \eqref{EQgeneral_matrix_time_node}. The basic idea is for each node $v_j$ to carefully track its trustworthy neighbors at iteration $k$, i.e., the nodes in the set $\mathcal{N}_j[k] = \mathcal{N}_j \cap \mathcal{T}_j[k]$, and to isolate neighbors that it does not consider trustworthy. This is done in two ways: (i) node $v_j$ ignores any values it receives at iteration $k$ from neighbors outside the set $\mathcal{N}_j[k]$, and (ii) node $v_j$ considers that its degree at iteration $k$ is $| \mathcal{N}_j[k] | = D_j[k]$. Effectively, node $v_j$ updates its value as 
\begin{eqnarray}
x_j[k+1] & = & \!\!\! \left ( 1 - \frac{D_j[k]}{N} \right) x_j[k] + \nonumber \\
 & & \!\!\!\!\!\!\! + \frac{1}{N} \!\!\!\! \sum_{v_i \in \mathcal{N}_j[k]} (\rho_{ji}[k+1] - \rho_{ji}[k]) + \varepsilon_j[k], \label{EQbasicadjustment} 
\end{eqnarray}
where $\varepsilon_j[k]$ is zero except when there are changes in the set $\mathcal{T}_j$. More specifically, if $\mathcal{T}_j[k] \neq \mathcal{T}_j[k-1]$ (i.e., some nodes that were previously considered trustworthy by node $v_j$ are now considered untrustworthy, and/or vice-versa), then node $v_j$ needs to make additional adjustments (only at iteration $k$) via a nonzero $\varepsilon_j[k]$. These adjustments are discussed in more detail later in this section.

With the above modification, it should be clear that the malicious nodes are excluded from the distributed computation, i.e., the trustworthy nodes ignore any transmissions they receive from them and do not consider them when computing their degree. Thus, if $\mathcal{T}_j[k] = \mathcal{V}_T$ for all $k$ (starting from $k=0$), then the above modified version of the average consensus algorithm effectively operates on the bidirectional communication graph $\mathcal{G}_{T} = (\mathcal{V}_T, \mathcal{E}_T)$ induced from $\mathcal{G}$ by restricting attention to the trustworthy nodes, which is assumed to be connected (Assumption~2). Therefore, the algorithm will converge to the average of the initial values of the trustworthy nodes comprising the graph $\mathcal{G}_{T}$, namely $\overline{X}_T$ given in \eqref{EQtrustaverage}. As we will see, this is effectively what iteration \eqref{EQbasicadjustment} will do; after $\mathcal{T}_j[k] = \mathcal{V}_T$ for $k \geq k_{\max}$ (see Assumption~3), trustworthy nodes are executing a distributed algorithm over $\mathcal{G}_T$.

The main challenge is that trust assessments may not be correct at the initialization of the algorithm and may also fluctuate during the operation of the algorithm. What is certain, thanks to Assumption~3, is that the trustworthy nodes will eventually exclude the malicious nodes from the distributed computation (and effectively operate on the induced graph $\mathcal{G}_{T}$ mentioned above). However, one needs to ensure that at the point this happens, the values maintained at the trustworthy nodes guarantee convergence to the trustworthy average. This is what we discuss next.

\subsection{Adding and Removing Malicious Nodes}

In order to perform the calculation of the trustworthy average, $\overline{X}_T$, trustworthy nodes need a way to add or remove the effects of trustworthy or malicious nodes on the distributed computation, using the available trust assessments at each iteration. This is mainly done by adjusting their degree and ignoring neighbors that are considered malicious (captured in the first two terms in the right hand side of \eqref{EQbasicadjustment}). In addition, each node needs to take specific actions when its perception about a neighboring node changes from malicious to trustworthy (or vice-versa), in order to ensure that the effect of that particular neighboring node is incorporated to (or removed from) the distributed computation; this is accomplished by having each (trustworthy) node $v_j$ properly adjust $\varepsilon_j[k]$ in \eqref{EQbasicadjustment}. There are two cases to consider, namely the case when a previously thought trustworthy neighbor $v_i$ of node $v_j$ becomes untrustworthy ($t_{ij}[k]=0$ and $t_{ij}[k-1] = 1$) and vice-versa.

\begin{itemize}

\item Case 1: Previously thought trustworthy neighbor $v_i$ becomes untrustworthy ($t_{ij}[k]=0$ and $t_{ij}[k-1] = 1$). In this case, apart from decreasing its degree by one ($D_j[k] = D_j[k-1]-1$) and ignoring node $v_i$'s transmissions, node $v_j$ needs to compensate for two quantities: 
\\ (i)  Node $v_j$ has to remove the effect of all previous transmissions that it received from node $v_i$. This can be done by subtracting from its $x_j$ value the amount $\frac{1}{N} \sigma_i[k]$ (note that $\sigma_i[k]$ is the running sum of node $v_i$, which node $v_j$ stores as $\rho_{ji}[k]$).
\\ (ii) Node $v_j$ needs to compensate for all transmissions that it has sent to node $v_i$ while it was considered trustworthy. It does so by adding to its $x_j$ value the amount $\frac{1}{N} \sigma_j[k]$. The reason is that $\sigma_j[k]$ represents the cumulative $x_j$ values that were sent to the now untrustworthy neighbor $v_i$ (or any neighbor for that matter). This adjustment essentially ``diverts" all messages sent to node $v_i$ from node $v_j$ up to (and including) time instant $k-1$, back to node $v_j$.
\\ The above two adjustments that node $v_j$ needs to make are as follows:
$$
\varepsilon_{ji}^-[k] = \frac{1}{N} \sigma_j[k] - \frac{1}{N} \rho_{ji}[k] \; .
$$

Note that if, at iteration $k$, node $v_j$ changes its perception about multiple neighboring nodes that were previously considered trustworthy, the above adjustments have to be performed for each such neighbor. In other words, if we let $\Delta \mathcal{U}_j[k] = \mathcal{N}_j \cap (\mathcal{T}_j[k-1] \setminus \mathcal{T}_j[k])$ be the set of previously trustworthy neighbors of node $v_j$ that become untrustworthy at iteration $k$, then, the total adjustment that node $v_j$ needs to make is as follows:
$$
\begin{array}{rcl}
\varepsilon_j^-[k] & = & \sum_{v_i \in \Delta \mathcal{U}_j[k]} \varepsilon_{ji}^-[k] \\ \\
 & = & | \Delta \mathcal{U}_{j}[k] | \frac{1}{N} \sigma_j[k] - \frac{1}{N} \sum_{v_i \in \Delta \mathcal{U}_j[k]} \rho_{ji}[k] \; .
\end{array}
$$

\item Case 2: Previously thought untrustworthy neighbor $v_i$ becomes trustworthy ($t_{ij}[k]=1$ and $t_{ij}[k-1] = 0$). In this case, apart from increasing its degree by one ($D_j[k] = D_j[k-1]+1$) and incorporating node $v_i$'s transmissions, node $v_j$ needs to compensate for two quantities: 
\\ (i)  Node $v_j$ has to add the effect of all previous transmissions that it has been ignoring from node $v_i$. This can be done by adding to $x_j$ the value $\frac{1}{N} \sigma_i[k]$ (note that $\sigma_i[k]$ is the running sum of node $v_i$, which node $v_j$ stores as $\rho_{ji}[k]$).
\\ (ii) Node $v_j$ needs to compensate for all transmissions that node $v_i$ was receiving from node $v_j$ while it was considered untrustworthy. Note that node $v_i$ was receiving the running sums of node $v_j$ without knowing that these transmissions were not intended for it. Thus, now that $v_i$ is back into the computation as a trustworthy node, node $v_j$ needs to adjust for this by subtracting from its $x_j$ value the amount $\frac{1}{N} \sigma_j[k]$. This adjustment essentially removes all the messages that node $v_i$ was receiving and incorporating in its calculation that should not have been incorporated (among other reasons because these values were computed by node $v_j$ using a degree that did not consider node $v_i$). 
\\ The above two adjustments that node $v_j$ needs to make are as follows:
$$
\varepsilon_{ji}^+[k] = - \frac{1}{N} \sigma_j[k] + \frac{1}{N} \rho_{ji}[k] \; .
$$

Note that if, at iteration $k$, node $v_j$ changes its perception about multiple neighboring nodes that were previously considered untrustworthy, the above adjustments have to be performed for each such neighbor. In other words, if we let $\Delta \mathcal{T}_j[k] = \mathcal{N}_j \cap (\mathcal{T}_j[k] \setminus \mathcal{T}_j[k-1])$ be the set of previously untrustworthy neighbors of node $v_j$ that become trustworthy at iteration $k$, then, the total adjustment that node $v_j$ needs to make is as follows:
$$
\begin{array}{rcl}
\varepsilon_j^+[k] & = & \sum_{v_i \in \Delta \mathcal{T}_j[k]} \varepsilon_{ji}^-[k] \\ \\
 & = & - | \Delta \mathcal{T}_{j}[k] | \frac{1}{N} \sigma_j[k] + \frac{1}{N} \sum_{v_i \in \Delta \mathcal{T}_j[k]} \rho_{ji}[k] \; .
\end{array}
$$

\end{itemize}

It is interesting to note that Cases~1 and~2 can be easily merged together as follows. At iteration $k$, node $v_j$ sets its $\varepsilon_j[k]$ value as $\varepsilon_j[k] = \varepsilon_j^-[k] + \varepsilon_j^+[k]$, i.e.,
\begin{eqnarray*}
\varepsilon_j[k] & = & (| \Delta \mathcal{U}_{j}[k] | - | \Delta \mathcal{T}_j[k] |) \frac{1}{N} \sigma_j[k] - \\
 & & ~~ - \frac{1}{N} \sum_{v_i \in \Delta \mathcal{U}_j[k]} \rho_{ji}[k] + \frac{1}{N} \sum_{v_i \in \Delta \mathcal{T}_j[k]} \rho_{ji}[k] \; ,
\end{eqnarray*} 
and then updates $x_j[k+1]$ following \eqref{EQbasicadjustment}:
$$
\begin{array}{rcl}
x_j[k+1] & = & \left ( 1 - \frac{D_j[k]}{N} \right ) x_j[k] + \\ 
 & & + \frac{1}{N} \sum_{v_i \in \mathcal{N}_j[k]} (\rho_{ji}[k+1] - \rho_{ji}[k]) + \varepsilon_j[k] \; ,
\end{array}
$$
where $\rho_{ji}[k+1] = \sigma_i[k+1]$.

In our analysis and the pseudocode of Algorithm~1, we use an equivalent but simpler way to perform the updates. More specifically, at iteration $k$, node $v_j$ receives $\{ \sigma_i[k+1] \; | \; v_i \in \mathcal{N}_j \}$ and sets $\mu_{ji}$ for each neighbor $v_i \in \mathcal{N}_j$ as follows:
$$
\mu_{ji}[k+1]  = \left \{ 
\begin{array}{ll} 
\sigma_i[k+1] \;, & \forall v_i \in \mathcal{N}_j[k] \; , \\
0 \;, & \text{otherwise;}
\end{array} \right .
$$
then node $v_j$ sets 
$$
\begin{array}{rcl}
x_j[k+1] & = & \left ( 1 - \frac{D_j[k]}{N} \right ) x_j[k] + \\
 & & + \frac{1}{N} \sum_{v_i \in \mathcal{N}_j} (\mu_{ji}[k+1] - \mu_{ji}[k]) + e_j[k]
 \end{array}
$$
where $e_j[k]$ is given by 
$$
e_j[k] = (| \Delta \mathcal{U}_{j}[k] | - | \Delta \mathcal{T}_j[k] |) \frac{1}{N} \sigma_j[k] \; .
$$
Note that if neighbor $v_i$ becomes untrustworthy for the first time at iteration $k$, $\mu_{ji}[k+1]$ becomes zero, so that the summation above effectively subtracts $\mu_{ji}[k] = \sigma_i[k]$. In addition, as long as neighbor $v_i$ remains untrustworthy to node $v_j$, its $\mu_{ji}$'s are zero (thus, the difference between two consecutive $\mu_{ji}$ is also zero, i.e., $v_i$ is effectively ignored in the summation). Finally, if $v_i$ becomes trustworthy again at some iteration $k'$, $\mu_{ji}[k'+1] = \sigma_i[k'+1]$, which means that the running sum is added back into the computation (since $\mu_{ji}[k'] = 0$).

\begin{remark}
It is worth pointing out that all of the above is described from the perspective of node $v_j$ based on its own trust assessments (i.e., its own sets of trustworthy nodes $\mathcal{T}_j[k]$ and $\mathcal{T}_j[k-1]$). All other (trustworthy) nodes are assumed to follow an identical procedure based on their own trust assessments. Finally, note that malicious nodes are allowed to behave arbitrarily.
\end{remark}

\begin{algorithm}[t!]
  \begin{small}
	\NoCaptionOfAlgo
	\SetKwBlock{Begin}{Compute:}{end}
	\nl \KwIn{Node $v_j$ knows $x_j$, $\mathcal{N}_j$, and has access to sets $\mathcal{T}_j[k]$ for $k=0,1,2,...$}
	\SetKwFor{While}{Initialization:}{}{endw}
	\While{}{
	\nl Node $v_j$ initializes $\mathcal{T}_j[-1] = \mathcal{V}$, $\mathcal{N}_j[-1] = \mathcal{N}_j$, \\ 
	$x_j[0] = x_j$, $\sigma_j[0] = 0$, and $\mu_{ji}[0] =0$, $\forall v_i \in \mathcal{N}_j$}
	\SetKwFor{For}{for}{}{endfor}
	\For{$k \geq 0:$}{
		\BlankLine
         	  \nl  Receive $\mathcal{T}_j[k]$ \\
		\SetKwFor{While}{Update Sets of Trustworthy Neighbors:}{}{endw}
		\While{}{
%		          \nl  Set $\mathcal{N}_j^-[k] = \mathcal{N}_j^- \cap \mathcal{T}_j[k]$ \\
		          \nl  Set $\mathcal{N}_j[k] = \mathcal{N}_j \cap \mathcal{T}_j[k]$ \\
		          \nl  Set $D_j[k] = | \mathcal{N}_j[k] |$ \\
%		          \nl  Set $D_j^+[k] = | \mathcal{N}_j^+[k] |$ \\
		          \nl  Set $\Delta \mathcal{U}_j[k] = \mathcal{N}_j \cap (\mathcal{T}_j[k-1] \setminus \mathcal{T}_j[k])$ \\ 
		          \nl  Set $\Delta \mathcal{T}_j[k] = \mathcal{N}_j \cap (\mathcal{T}_j[k] \setminus \mathcal{T}_j[k-1])$ \\
%		          \nl Cases~1 and~2: \\
%                          \nl $\tilde{z}_j[k] := z_j[k] + (\Delta U^+_j[k]-\Delta T^+_j[k]) \eta_j[k]$ \\ ~\\ 
%                          \nl Case~3 and~4: \\
%                          \nl Set $\mathcal{N}_{un,j}^-[k] = \mathcal{N}_j^- \cap (\mathcal{T}_j[k-1] \setminus \mathcal{T}_j[k])$ \\
%                          \nl $y_j[k] := y_j[k] - \sum_{v_i \in \mathcal{N}_{un,j}^-[k]} \rho_{ji}[k]$ \\
%			 \nl $z_j[k] := z_j[k] - \sum_{v_i \in \mathcal{N}_{un,j}^-[k]} \nu_{ji}[k]$ \\ ~\\
%			 \nl Case 4: \\
%			 \nl For each $v_i \in \mathcal{N}_{un,j}^-[k]$, set $\rho_{ji}[k] = 0$
		}
		
		\SetKwFor{While}{Compute:}{}{endw}
		\While{}{
			  \nl $   \sigma_j[k+1]=\sigma_j[k] + x_j[k]$ \label{e_sigma0} \\
% 	 		  \nl $   \eta_j[k+1]=\eta_j[k]+{\tilde{z}_j[k]}/{(1+D_j^+[k])}$ \label{e_eta0} \\
		}

		\BlankLine
		\nl\textbf{Broadcast:} $\sigma_j[k+1]$ to all $v_l \in \mathcal{N}_j$ 
		\BlankLine
		\nl\textbf{Receive:} $\sigma_i[k+1]$ from each $v_i \in \mathcal{N}_j$ 

		\BlankLine
				\SetKwFor{While}{Update $\mu_{ji}$'s:}{}{endw}
		\While{}{
%		          \nl  Set $\mathcal{N}_j^-[k] = \mathcal{N}_j^- \cap \mathcal{T}_j[k]$ \\
%		          \nl Cases~3 and~4: \\
                           \nl For each $v_i \in \mathcal{N}_j$ set
                           \nl $
\mu_{ji}[k+1]  = \left \{ 
\begin{array}{ll} 
\sigma_i[k+1] \;, & \forall v_i \in \mathcal{N}_j[k] \\
0 \;, & \text{otherwise}
\end{array} \right .
$
		}
    		\BlankLine
		\SetKwFor{While}{Compute:}{}{endw}
		\While{}{
       		\nl $e_j[k] := (| \Delta \mathcal{U}_j[k] | - | \Delta \mathcal{T}_j[k] |) \frac{1}{N} \sigma_j[k]$ \\ ~\\
%		\nl  $\rho_{ji}[k+1]  = \sigma_i[k+1] $, $\forall v_i \in \mathcal{N}_j^-[k] \cup \{ v_j \}$ (otherwise, $\rho_{ji}[k+1] = 0$) \\ 
          	\nl $x_j[k+1] = \left (1 - \frac{D_j[k]}{N} \right ) x_j[k] +$ \\ ~\\
	        \nl ~~~~~~~~~~~~~~~~~~~~$+ \frac{1}{N} \sum_{v_i \in \mathcal{N}_j} (\mu_{ji}[k+1] - \mu_{ji}[k]) + e_j[k]$ \\ ~\\
% 	  \sum_{v_i \in \mathcal{N}_j^- \cup \{ v_j \}} (\rho_{ji}[k+1] - \rho_{ji}[k])$ \\ ~\\
 % 		\nl  $\nu_{ji}[k+1]  = \eta_i[k+1] $, $\forall v_i \in \mathcal{N}_j^-[k] \cup \{ v_j \}$ (otherwise, $\nu_{ji}[k+1] = 0$) \\ 
 %         	\nl $z_j[k+1] =   \sum_{v_i \in \mathcal{N}_j^- \cup \{ v_j \}} (\nu_{ji}[k+1] - \nu_{ji}[k])$ 
		 %\nl
		}
				
%		\BlankLine
%		\nl\textbf{Broadcast:} $\sigma_j[k+1]$ and $\nu_j[k+1]$ to all $l \in \mathcal{N}_j^+$ \label{out_weight_alg0}
	}
	\caption{\textbf{Algorithm~1: Trustworthy Distributed Averaging}}
	\label{algorithm_0}
	% \vspace{-0.1in}
	\end{small}
\end{algorithm}

\subsection{Proof of Correctness} % Algorithm}
\label{SUBSECpseudocode}

Algorithm~1 is presented from the perspective of node $v_j$. Each node $v_j$ is assumed to know the set of its neighbors $\mathcal{N}_j$ and receives at each iteration $k$ (or can determine based on trust measurements) its set of trustworthy nodes $\mathcal{T}_j[k]$. Node $v_j$ computes the set of trustworthy neighbors at iteration $k$ as $\mathcal{N}_j[k] = \mathcal{N}_j \cap \mathcal{T}_j[k]$ and its degree as $D_j[k] = | \mathcal{N}_j[k] |$. To establish the main convergence result, we first state an important invariant that holds during the execution of Algorithm~\ref{algorithm_0}.

\begin{theorem}
\label{THEinvariant}
Consider a distributed system, captured by a bidirectional communication graph $\mathcal{G} = (\mathcal{V}, \mathcal{E})$, in which each node $v_j \in \mathcal{V}$ has value $x_j$. A certain subset $\mathcal{V}_T$, $\mathcal{V}_T \subseteq \mathcal{V}$, of the nodes are trustworthy (non-malicious), whereas the remaining nodes in $\mathcal{V}_M = \mathcal{V} \setminus \mathcal{V}_T$ are untrustworthy (malicious).  Consider the execution of Algorithm~1 where at each iteration~$k$, each node $v_j$ has a binary indicator $t_{ij}[k]$ regarding the trust it places to node $v_i$. Let $\mathcal{N}_j[k] = \mathcal{N}_j\cap \mathcal{T}_j[k]$, where $\mathcal{T}_j[k] = \{ v_i \; | \; t_{ij}[k]=1 \}$ is the set of nodes that are considered trustworthy by node $v_j$ at iteration $k$ (i.e., $\mathcal{N}_j[k]$ is the set of neighbors that are considered trustworthy by node $v_j$ at iteration $k$). Under Assumptions~0--2,\footnote{Note that Assumption~3 is not really needed for the invariant to hold; in fact, Assumption~1 can also be relaxed.} for each trustworthy node $v_j \in \mathcal{V}_T$, it holds at each iteration $k$ ($k = 0, 1, ...$)
\begin{equation}
x_j[k] - x_j = - \frac{D_j[k-1]}{N} \sigma_j[k] + \frac{1}{N} \sum_{v_i \in \mathcal{N}_j[k-1]} \sigma_i[k] \; , \label{EQinvariant}
\end{equation}
where $D_j[k] = | \mathcal{N}_j[k] |$ is the number of trustworthy neighbors of node $v_j$ at iteration $k$ (recall that, according to the initialization of Algorithm~\ref{algorithm_0}, we have $D_j[-1] = | \mathcal{N}_j[-1] | = D_j$).
\end{theorem}

\begin{proof}
We prove the invariant by induction on $k$. At $k=0$, the invariant clearly holds since $x_j[0]=x_j$ and all $\sigma_i[0] = 0$. Notice that the exact $\mathcal{T}_j[-1]$ and $D_j[-1]$ are not really relevant here (because all $\sigma_i[0]=0$) but Algorithm~\ref{algorithm_0} sets them to $\mathcal{T}_j[-1] = \mathcal{V}$ and $D_j[-1]=D_j$.

Suppose that at $k=t$ the invariant holds, i.e.,
\begin{equation}
x_j[t] - x_j = - \frac{D_j[t-1]}{N} \sigma_j[t]  + \frac{1}{N} \sum_{v_i \in \mathcal{N}_j[t-1]} \sigma_i[t] \; . \label{EQproofinvariant}
\end{equation}
At iteration $k=t+1$, we need to show that 
$$
x_j[t+1] - x_j = - \frac{D_j[t]}{N} \sigma_j[t+1]  + \frac{1}{N} \sum_{v_i \in \mathcal{N}_j[t]} \sigma_i[t+1] \; .
$$

From lines~8 and~14 of Algorithm~\ref{algorithm_0}, we have the following:
$$
\begin{array}{rcl}
\sigma_j[t+1] & = & \sigma_j[t] + x_j[t] \; , \\ \\
x_j[t+1]  & = & \left ( 1 - \frac{D_j[t]}{N} \right ) x_j[t] + \\ \\
 & & ~ + \frac{1}{N} \sum_{v_i \in \mathcal{N}_j} (\mu_{ji}[t+1] - \mu_{ji}[t]) + e_j[t] \; , 
\end{array}
$$
where $e_j[t] = (| \Delta \mathcal{U}_{j}[t] | - | \Delta \mathcal{T}_j[t] |) \frac{1}{N} \sigma_j[t]$ (line~13) and $\mu_{ji}[t+1] = \sigma_i[t+1]$ if $v_i \in \mathcal{N}_j[t]$, otherwise, $\mu_{ji}[t+1] = 0$ (line~12).

Substituting $e_j[t]$ and using the iteration invariant in \eqref{EQproofinvariant}, we have 
$$
\begin{array}{rcl}
x_j[t+1]  & = & \underbrace{x_j - \frac{D_j[t-1]}{N} \sigma_j[t] + \frac{1}{N} \sum_{v_i \in \mathcal{N}_j[t-1]} \sigma_i[t]}_{=x_j[t] \text{ by \eqref{EQproofinvariant}}} \\ \\
              &    & - \frac{D_j[t]}{N} x_j[t] + \frac{1}{N} \sum_{v_i \in \mathcal{N}_j} (\mu_{ji}[t+1] - \mu_{ji}[t]) \\ \\
	      &    & + (| \Delta \mathcal{U}_{j}[t] | - | \Delta \mathcal{T}_j[t] |) \frac{1}{N} \sigma_j[t] \; .
\end{array}
$$
Rearranging terms, we have 
$$
\begin{array}{l}
x_j[t+1] - x_j = \\
~ = - \frac{1}{N} ( D_j[t-1] - | \Delta \mathcal{U}_{j}[t] | + | \Delta \mathcal{T}_j[t] |) \sigma_j[t] - \frac{D_j[t]}{N} x_j[t] + \\ \\
~ + \frac{1}{N} \sum_{v_i \in \mathcal{N}_j} (\mu_{ji}[t+1] - \mu_{ji}[t]) + \frac{1}{N} \sum_{v_i \in \mathcal{N}_j[t-1]} \sigma_i[t] \; . 
\end{array}
$$

We observe that $D_j[t-1] - | \Delta \mathcal{U}_{j}[t] | + | \Delta \mathcal{T}_j[t]| = D_j[t]$ and that $\sigma_j[t]+x_j[t]=\sigma_j[t+1]$; thus, 
\begin{eqnarray}
x_j[t+1] - x_j & = & - \frac{1}{N} D_j[t] \sigma_j[t+1] + \frac{1}{N} \sum_{v_i \in \mathcal{N}_j[t-1]} \sigma_i[t] \nonumber \\ \nonumber \\ 
& & + \frac{1}{N} \sum_{v_i \in \mathcal{N}_j} (\mu_{ji}[t+1] - \mu_{ji}[t]) \; . \label{EQinter1}
\end{eqnarray}
Furthermore, since $\mu_{ji}[t+1] = \sigma_i[t+1]$ if $v_i \in \mathcal{N}_j[t]$ (otherwise, $\mu_{ji}[t+1] = 0$), we can write 
\begin{eqnarray}
\frac{1}{N} \sum_{v_i \in \mathcal{N}_j} (\mu_{ji}[t+1] - \mu_{ji}[t]) & = & \nonumber \\
 = \frac{1}{N} \sum_{v_i \in \mathcal{N}_j[t]} \sigma_{i}[t+1] & - & \frac{1}{N} \sum_{v_i \in \mathcal{N}_j[t-1]} \sigma_{i}[t] \; . ~~~~~ \label{EQinter2}
\end{eqnarray}

Then, by using \eqref{EQinter2} together with \eqref{EQinter1}, we obtain
$$
x_j[t+1] - x_j = - \frac{D_j[t]}{N} \sigma_j[t+1] + \frac{1}{N} \sum_{v_i \in \mathcal{N}_j[t]} \sigma_i[t+1] \; ,
$$
which establishes the proof of the invariant. \qed
\end{proof}

The main convergence result for Algorithm~1 is stated next.

\begin{theorem}
\label{THEmain}
Consider a distributed system, captured by a bidirectional communication graph $\mathcal{G} = (\mathcal{V}, \mathcal{E})$, in which each node $v_j \in \mathcal{V}$ has an initial value $x_j$. A certain subset $\mathcal{V}_T$, $\mathcal{V}_T \subseteq \mathcal{V}$, of the nodes are trustworthy (non-malicious), whereas the remaining nodes in $\mathcal{V}_M = \mathcal{V} \setminus \mathcal{V}_T$ are untrustworthy (malicious). At various points in time, each node $v_j$ receives information that allows it to compute a binary indicator $t_{ij}$ regarding the trust it places to node $v_i$. Under Assumptions~$0$--$3$, if the non-malicious nodes execute Algorithm~1, they asymptotically converge to the average of their initial values  given by $\overline{X}_T = \frac{\sum_{v_l \in \mathcal{V}_T} x_l}{|\mathcal{V}_T|}$.
\end{theorem}

\begin{proof}
Consider a large $k_0$, such that all $t_{ij}[k]$ for $k \geq k_0$ (for all $v_j \in \mathcal{V}_T$) correctly reflect the status of node $v_i$ from the point view of node $v_j$. Such a $k_0$ exists under Assumption~3. Effectively, from this iteration onwards, the nodes are running a standard average consensus algorithm on the induced subgraph $\mathcal{G}_T$ that only involves the trustworthy nodes and is connected under Assumption~2. Therefore, the values $x_j[k]$ of all trustworthy nodes will converge to 
\begin{equation}
 \lim_{k\rightarrow\infty} x_j[k]=\frac{\sum_{v_l \in \mathcal{V}_T} x_l[k_0+1]}{ | \mathcal{V}_T |} \; .
\end{equation}
We next argue that $\sum_{v_l \in \mathcal{V}_T} x_l[k_0+1] = \sum_{v_l \in \mathcal{V}_T} x_l$, which establishes the result.

Using the invariant in Theorem~\ref{THEinvariant} at time $k_0+1$, we have that, for each $v_j \in \mathcal{V}_T$, it holds that 
$$
x_j[k_0+1] - x_j = - \frac{D_j[k_0]}{N} \sigma_j[k_0+1] + \frac{1}{N} \sum_{v_i \in \mathcal{N}_j[k_0]} \sigma_i[k_0+1] \; ,
$$
where $D_j[k_0]$ ($\mathcal{N}_j[k_0]$) is the degree (set of neighbors) of node $v_j$ in the graph $\mathcal{G}_{T}$. If we sum all such equations over all $v_j$ in $\mathcal{V}_T$, on the left, we have the sum $\sum_{v_j \in \mathcal{V}_T} (x_j[k_0+1] - x_j)$, whereas on the right we have 
$$
\begin{array}{l}
- \frac{1}{N} \sum_{v_j \in \mathcal{V}_T} (D_j[k_0] \sigma_j[k_0+1]) + \\
~~~~ \frac{1}{N} \sum_{v_j \in \mathcal{V}_T} \left ( \sum_{v_i \in \mathcal{N}_j[k_0]} \sigma_i[k_0+1] \right ) = 0 \; ,
\end{array}
$$
where the simplification occurs because, for each link $(v_j, v_i)$ of $\mathcal{G}_T$, we have $\sigma_j$ appearing twice, once with a positive sign and once with a negative sign. 

This establishes that $\sum_{v_j \in \mathcal{V}_T} (x_j[k_0+1] - x_j) = 0$, i.e., $\sum_{v_j \in \mathcal{V}_T} x_j[k_0+1] = \sum_{v_j \in \mathcal{V}_T} x_j$, which essentially proves the theorem. \qed
\end{proof}

\section{Numerical Simulations}
\label{SECexamples}

\subsection{5-Node Network}

Consider the bidirectional communication graph $\mathcal{G} = (\mathcal{V}, \mathcal{E})$ in Fig.~\ref{FIGsmallexample}, with nodes $\mathcal{V}= \{ v_1, v_2, v_3, v_4, v_5 \}$ and edges $\mathcal{E} = \{ (v_1, v_2), (v_2, v_1), ..., (v_4, v_5), (v_5, v_4) \}$ as indicated in the figure. We assume that the set of trustworthy nodes is $\mathcal{V}_T = \{ v_1, v_2, v_3, v_4 \}$ and the set of malicious nodes is $\mathcal{V}_M = \{ v_5 \}$. Notice that the induced graph $\mathcal{G}_{T} = (\mathcal{V}_T, \mathcal{E}_T)$ (where $\mathcal{E}_T = \{ (v_2, v_1), (v_1, v_2), (v_2, v_3), (v_3, v_2), (v_3, v_4),$ $(v_4, v_3), (v_4, v_1), (v_1, v_4) \}$) is connected. We assume that the initial values of the nodes are $x_i = i$ for $i=1, 2, ..., 5$, so that the average is $\overline{X}=3$ and the average of the trustworthy nodes is $\overline{X}_T = 2.5$. Throughout this example, we set the edge weights equal to $\frac{1}{5}$.

\begin{figure}
    \centering
    \includegraphics[width=0.5\linewidth]{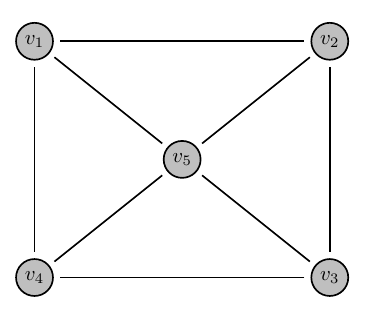}
\caption{Bidirectional communication graph considered in the 5-node example.}
\label{FIGsmallexample}
\end{figure}

\begin{figure*}
    \centering
    \includegraphics[width=0.33\linewidth]{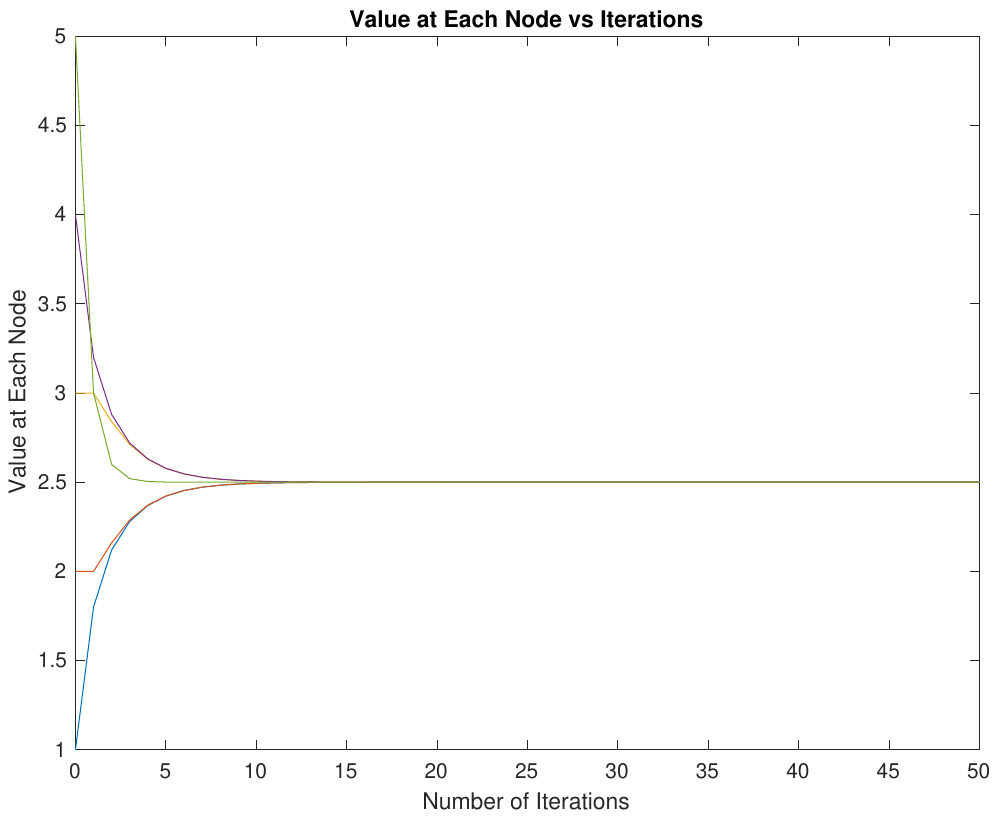}~\includegraphics[width=0.33\linewidth]{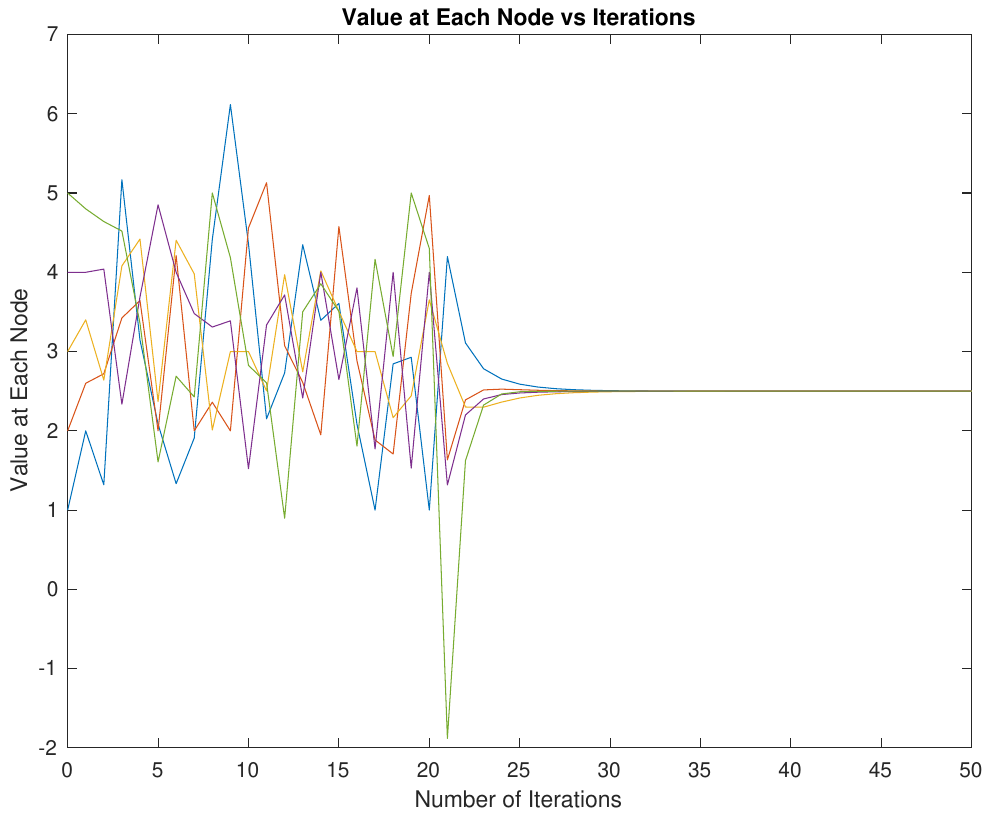}~\includegraphics[width=0.33\linewidth]{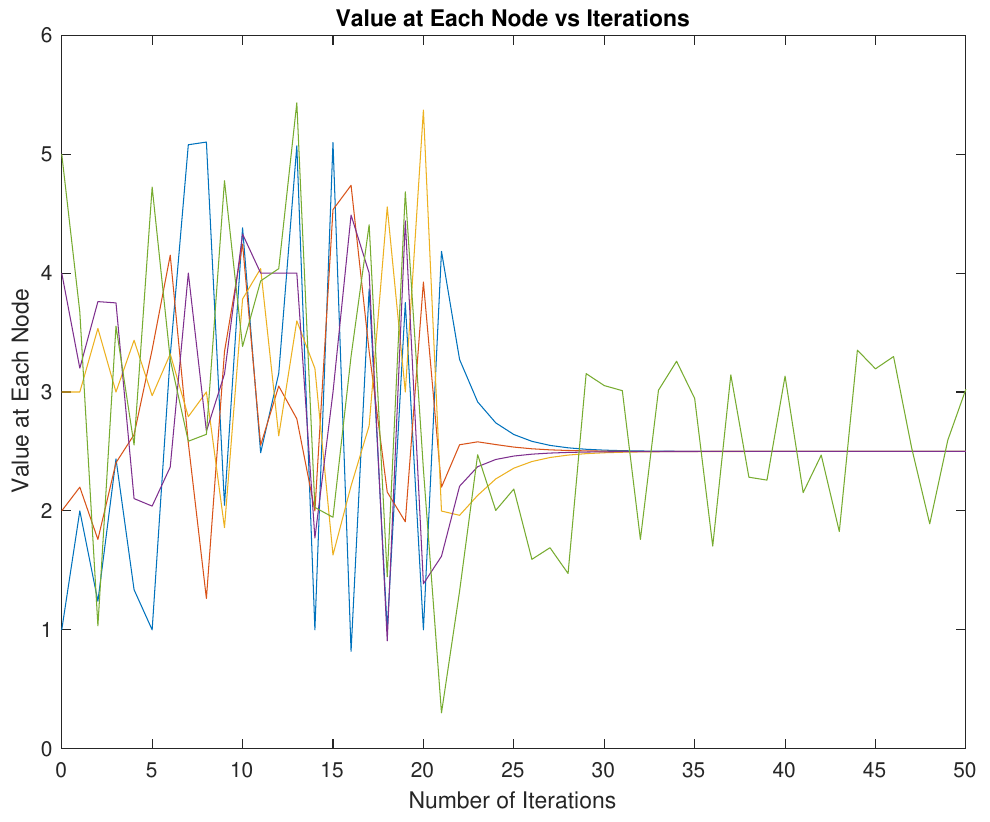}

\caption{The values of nodes in the 5-node example converge to the average $\overline{X}_T=2.5$ when trust assessments are taken into account: correct trust assessments from iteration $k=0$ onwards, with malicious node $v_5$ behaving correctly (left); correct trust assessments from iteration $k=21$ onwards (random trust assessments before $k=21$), with malicious node $v_5$ behaving correctly (middle); correct trust assessments from iteration $k=21$ onwards (random trust assessments before $k=21$), with malicious node $v_5$ behaving incorrectly (right).}
\label{FIGmalicious}
\end{figure*}

We next execute different scenarios to illustrate the operation of the proposed trustworthy distributed averaging algorithm. More specifically, we illustrate three runs of the proposed algorithm, which differ in terms of when trust assessments converge to the correct values and/or the behavior of the malicious node. In all three cases, the non-malicious nodes in $\mathcal{V}_T$ converge to the average of the trustworthy nodes $\overline{X}_T=2.5$, whereas the malicious node may or may not converge depending on its own behavior.
% \begin{itemize}
% \item 

On the left of Fig.~\ref{FIGmalicious}, we see the behavior of the network when nodes perceive node $v_5$ as untrustworthy from the very beginning. In this simulation, node $v_5$ behaves normally (despite the fact that all other nodes perceive it as untrustworthy) and we see that it also converges to the average of the trustworthy nodes.
% \item
In the middle of Fig.~\ref{FIGmalicious}, we see the behavior of the network when, up to iteration~20, nodes receive randomly generated binary values for the trust assessments of other nodes; however, after iteration 20, trust assessments settle. We see that the trustworthy nodes converge to the average $\overline{X}_T$ after about 30 iterations. In this simulation, we also assume that the malicious node $v_5$ behaves normally and we see that it also converges to the average $\overline{X}_T$. 
% \item 
On the right of Fig.~\ref{FIGmalicious}, we see a simulation with the same characteristics as the one in the middle, except that the malicious node $v_5$ behaves arbitrarily (more specifically, at each iteration, node $v_5$ adds a random offset to its $x$ value, which then propagates to the values it transmits to its neighbors). In this case, the malicious node does not converge to a value; however, the trustworthy nodes are able to converge to the average $\overline{X}_T$ after about 30 iterations. Moreover, the transmissions of the malicious node stop influencing the distributed computation, including any offsets it added before iteration~$k=21$.
% \end{itemize}

\subsection{20-Node Network}

\begin{figure*}
    \centering
    \includegraphics[width=0.33\linewidth]{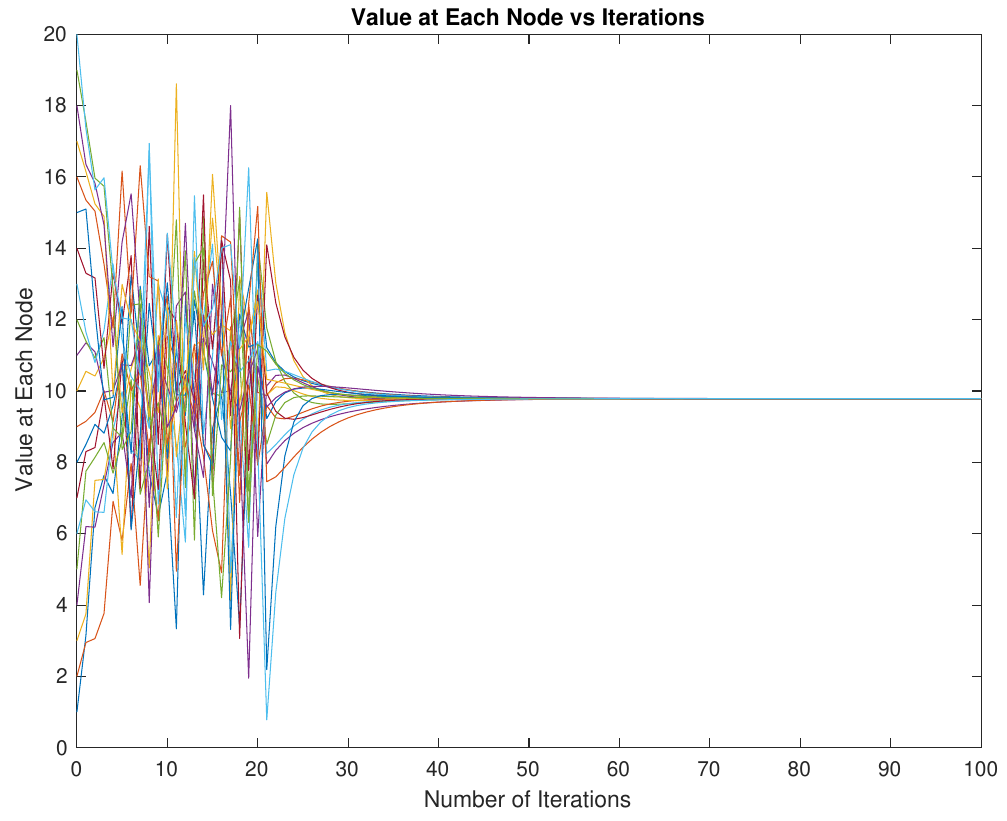}~\includegraphics[width=0.33\linewidth]{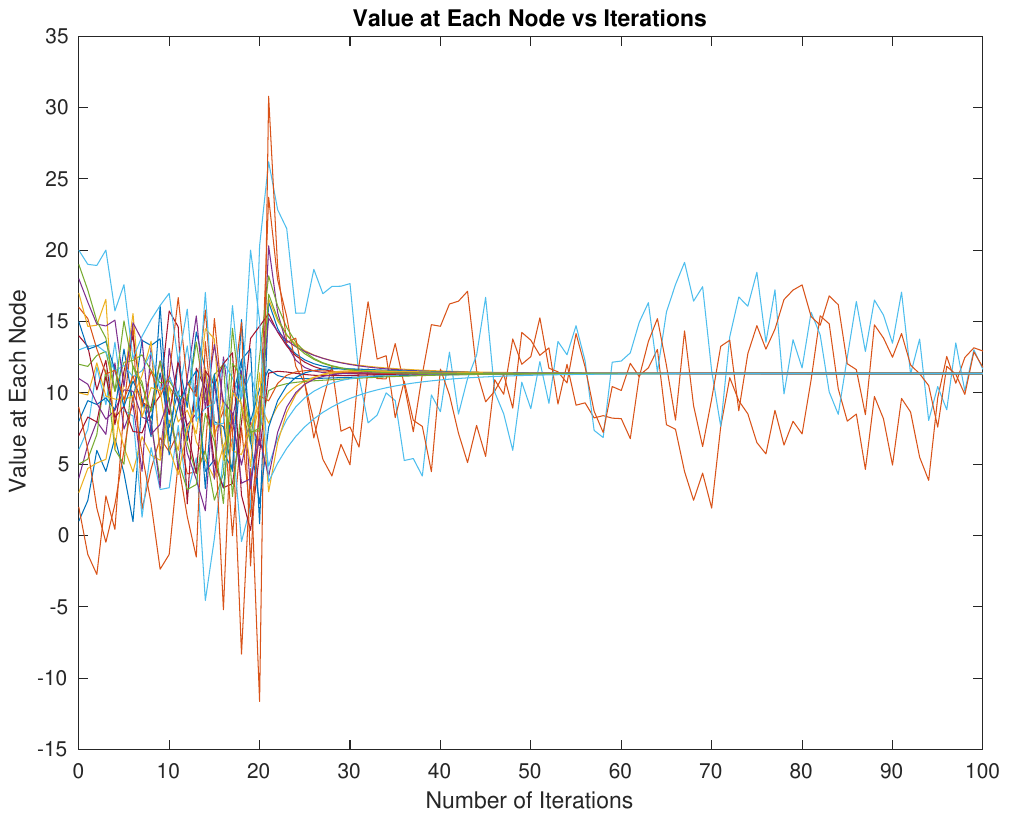}~\includegraphics[width=0.33\linewidth]{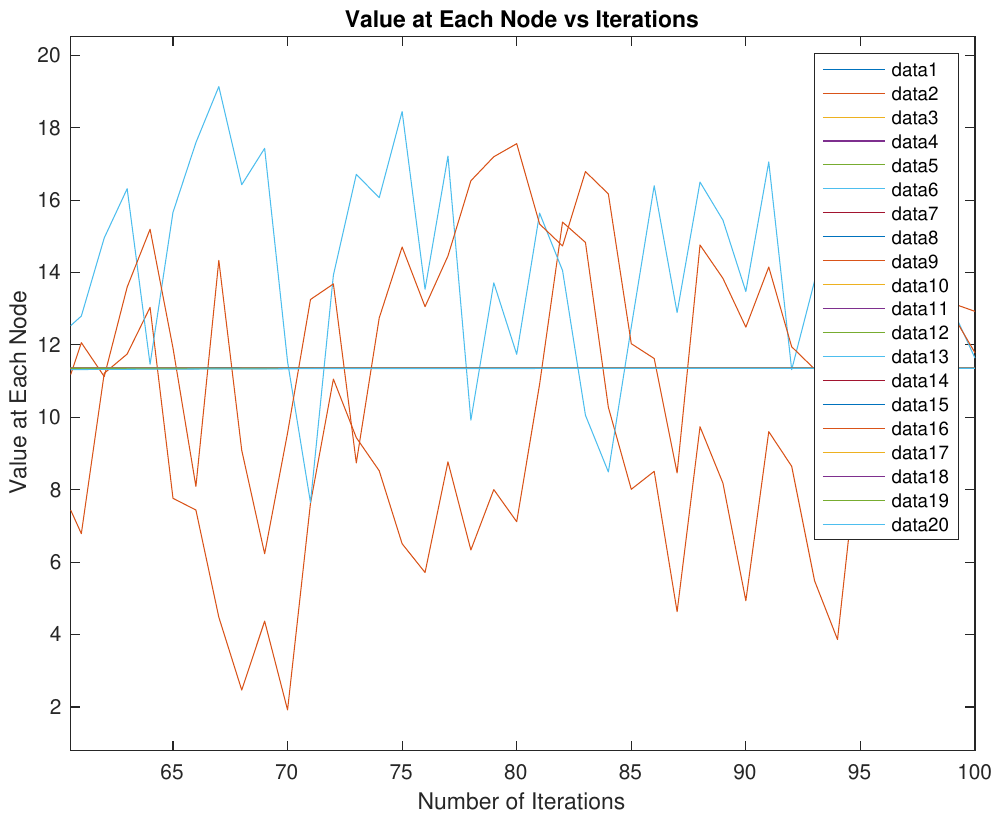}

\caption{The values of non-malicious nodes in the large examples converge to the average $\overline{X}_T$ when trust assessments are taken into account: correct trust assessments from iteration $k=21$ onwards  (random trust assessments before $k=21$), with malicious nodes $\mathcal{V}_M = \{ v_6, v_8, v_{11}, v_{14}, v_{15}, v_{19} \}$ behaving correctly (left); correct trust assessments from iteration $k=21$ onwards (random trust assessments before $k=21$), with malicious nodes $\mathcal{V}_M = \{ v_2, v_6, v_9 \}$ behaving incorrectly (middle); zoomed-in version of plot in the middle, which clearly shows that nodes in $\mathcal{V}_M$ do not converge (right).}
\label{FIGmalicious2}
\end{figure*}

In this section, we present two simulations with larger (randomly generated) connected bidirectional communication graphs that consist of $20$ nodes ($\mathcal{V}= \{ v_1, v_2, ..., v_{20} \}$) with initial values $x_i = i$ for $i=1, 2, ..., 20$, so that $\overline{X} = 10.5$. In both simulations, nodes have randomly generated binary values for their trust assessments about other nodes up to iteration~20; however, after iteration 20, they acquire correct trust assessments. In Fig.~\ref{FIGmalicious2}, we see that the non-malicious nodes in $\mathcal{V}_T$ converge to the average $\overline{X}_T$, whereas malicious nodes may or may not converge depending on their own behavior. Specifically, on the left of Fig.~\ref{FIGmalicious2}, we see the behavior of the network when the set of malicious nodes is $\mathcal{V}_M = \{ v_6, v_8, v_{11}, v_{14}, v_{15}, v_{19} \}$ and the malicious nodes behave correctly. In this case, all nodes converge to the average of the trustworthy nodes, which is $\overline{X}_T = 9.7857$, after about $60$ iterations. In the middle of Fig.~\ref{FIGmalicious2}, we see the behavior of the network when the set of malicious nodes is $\mathcal{V}_M = \{ v_2, v_6, v_9 \}$ and the malicious nodes behave arbitrarily (more specifically, at each iteration, each node in $\mathcal{V}_M$ adds a random offset to its $x$ value, which then propagates to the values it transmits to its neighbors). Again, we see that the non-malicious nodes converge to the average of the trustworthy nodes, which is $\overline{X}_T = 11.3529$, after about $60$ iterations. However, the malicious nodes do not converge, as seen clearly in the zoomed-in plot on the right of Fig.~\ref{FIGmalicious2}.

\section{Distributed Trust Evaluation Protocol}
\label{SECtrust}

The approach described in the earlier sections relies on the availability of trust assessments but it is independent of how such trust assessments are obtained (as long as they satisfy Assumption~3). In the literature, there are many proposals for obtaining trust assessments, including the schemes in \cite{yemini2021characterizing,agkun2022} described earlier in the paper. Inspired by the work in \cite{yuan2019resilient, yuan2021resilient, yuan2021secure}, which considers bidirectional communication graphs where each node is in charge of verifying the proper functionality of each of its neighbors by having access to two-hop information (i.e., by having access to information sent by the neighbors of its neighbors), we propose in this section a scheme that allows nodes that are running Algorithm~1 to assess the trustworthiness of each of their neighbors. This scheme can be embedded in the iterations of Algorithm~\ref{algorithm_0}, effectively allowing nodes to obtain the needed trust assessments.

Initially, all nodes consider their neighbors to be trustworthy (thus, their initial value is included in the average computation); however, if a node attempts to alter the outcome of the average computation (by calculating and transmitting incorrect values), then it should be declared malicious and its initial value should also be removed from the computation of the trustworthy average. Thus far in the paper, we did not have to explicitly define what constitutes a malicious node (since that was seamlessly provided by the trust assessments under Assumption~3). To obtain the trust assessments in this section, we define a malicious node as follows.\footnote{A subtle difference from the setting in the previous section is that, unless a node performs an incorrect update, it is not considered malicious and its initial value is included in the average calculation.}

\begin{definition}
A node $v_i$ executing Algorithm~\ref{algorithm_0} is malicious if, at any iteration $k$, it provides incorrect values $\sigma_i[k+1]$ to its neighbors (or computes incorrect values $x_j[k+1]$ which will inevitably alter $\sigma_i$ at later steps).
\end{definition}

Unlike \cite{yuan2019resilient, yuan2021resilient, yuan2021secure}, the approach proposed in this section does not require a separate (centralized) mechanism that allows all nodes to instantly learn who has been declared malicious, because Algorithm~1 requires each trustworthy node to only have trust assessments about its neighbors and, as we will see, such trust assessments become directly available to trustworthy nodes under the proposed scheme. Moreover, due to the ability of Algorithm~\ref{algorithm_0} to completely remove the effect of earlier exchange of information with malicious nodes, the trust assessments in the proposed protocol can be finalized over several time steps. Finally, as we argue at the end of this section, the invariant that we identified and established in Theorem~\ref{THEinvariant} can be used to perform such checking infrequently (more precisely, at random time instants) and not necessarily at each time step as done in \cite{yuan2019resilient, yuan2021resilient, yuan2021secure}; this is an important feature as it significantly relaxes the two-hop communication overhead imposed by the necessity to obtain trust assessments.

Assumptions~0t and~2t below replace Assumptions~0 and~2 respectively. Assumption~3 is no longer necessary as we will explicitly describe how to obtain $\mathcal{T}_j[k]$ for each node.

\noindent
{\bf Assumption 0t.} Consider a distributed system whose topology is captured by a bidirectional communication graph $\mathcal{G} = (\mathcal{V}, \mathcal{E})$, where the following hold. (i) Each node $v_j$ is aware of the local topology around it up to two hops, i.e., it is aware of its neighbors and the neighbors of its neighbors, which we refer to as its two-hop neighbors and denote by $\mathcal{N}^{(2)}_j := \mathcal{N}_j \cup (\cup_{v_i \in \mathcal{N}_j} \mathcal{N}_i)$, as well as the interconnections among them. (ii) Each node $v_j$ is capable of two types of transmissions, broadcasting messages that are received by all of its neighbors in $\mathcal{N}_j$, and two-hop broadcasting messages that are received by all of its two-hop neighbors in $\mathcal{N}^{(2)}_j$. Furthermore, we assume that both types of transmissions are associated with a unique node ID that allows receiving nodes to identify the sending node.

\begin{remark}
Messages to two-hop neighbors can be sent in various ways, e.g., by transmitting at a higher power. This is likely to be more expensive and undesirable and is one of the reasons we propose, at the end of this section, a scheme that limits the use of such transmissions. An alternative way of thinking about two-hop transmissions is that we start with a dense network topology, but we divide the neighborhood of each node, {\em on purpose}, to one-hop and two-hop neighbors in order to provide a mechanism for trust assessments (and thus resilience to untrustworthy nodes). In other words, we carefully design the network topology to have the properties that we need for trust assessment evaluation.
\end{remark}

\noindent
{\bf Assumption 2t.} The bidirectional communication graph induced from $\mathcal{G}$ by restricting attention to the trustworthy nodes, denoted by $\mathcal{G}_{T} = (\mathcal{V}_T, \mathcal{E}_T)$, where $\mathcal{E}_T = \{ (v_j, v_i) \in \mathcal{E} \; | \; v_j, v_i  \in \mathcal{V}_T \}$, is connected. Furthermore, there are no malicious nodes that are neighbors, i.e., for $v_j, v_i \in \mathcal{V}_M$ (recall that $\mathcal{V}_M = \mathcal{V} \setminus \mathcal{V}_T$), we have that $(v_j, v_i) \notin \mathcal{E}$ and $(v_i, v_j) \notin \mathcal{E}$.

In the remainder of this section, unless we explicitly indicate otherwise, a ``transmission" or a ``broadcast" by node $v_j$ indicates that the message is sent to its immediate neighbors (in $\mathcal{N}_j$), whereas a ``two-hop transmission" or a ``two-hop broadcast" by node $v_j$ indicates that the message is sent to all of its two-hop neighbors (in $\mathcal{N}_j^{(2)}$), including its immediate neighbors. Furthermore, when we say that a node $v_i$ is declared malicious by another node $v_j$ at iteration $k$, we mean that $t_{ij}[t] = 0$ for all $t > k$.

\subsection{Two-Hop Information Received at Each Iteration}

Under Assumptions~0t,~1, and~2t, let us consider that, at initialization, each node $v_\ell \in \mathcal{V}$ sends to its neighbors its value $x_\ell[0]=x_\ell$, $\sigma_\ell[0]=0$ and $\mathcal{T}_\ell[-1] = \mathcal{V}$ (the last two are not really needed since its neighbors already expect these values). Subsequently, at the end of each iteration $k$ of Algorithm~1 ($k=0, 1, 2, ...$), each node $v_\ell \in \mathcal{V}$ sends the following information:
\begin{enumerate}
\item[A.] Node $v_\ell$ sends to its (immediate) neighbors in $\mathcal{N}_\ell$:
\begin{enumerate}
\item the set of its trustworthy neighbors $\mathcal{T}_\ell[k]$;
\item its updated value $x_\ell[k+1]$;
\end{enumerate}
\item [B.] Node $v_\ell$ also sends to all of its two-hop neighbors in $\mathcal{N}_\ell^{(2)}$ (including its immediate neighbors):
\begin{enumerate}
\item its updated running sum $\sigma_\ell[k+1]$ (this is sent anyway to the immediate neighbors of node $v_\ell$ when executing Algorithm~1, but here we require that this running sum is sent to all two-hop neighbors as well).
\end{enumerate}
\end{enumerate}

The above communication strategy is followed by all nodes, including each neighbor $v_i$ of node $v_j$ ($v_i \in \mathcal{N}_j$). Thus, at the end of each iteration $k$, each node $v_j$ can check the computations performed by each neighbor $v_i$, as it has access to all information that is needed from node $v_i$ to execute an iteration step. More specifically, node $v_j$ knows: (i) the running sums $\{ \sigma_{l}[k] \; | \; v_{l} \in \mathcal{N}_i \}$, obtained via the two-hop transmissions by nodes $v_l$, $v_l \in \mathcal{N}_i$, at iteration $k-1$ ($v_l \in \mathcal{N}_i$, thus $v_l \in \mathcal{N}_j^{(2)}$); (ii) the value $x_{i}[k]$, obtained via the one-hop transmission by node $v_i$ at iteration $k-1$; (iii) the value $\sigma_{i}[k]$, obtained via the two-hop transmission by node $v_i$ at iteration $k-1$ (this information can be obtained via a one-hop transmission from node $v_i$ to node $v_j$, but it is sent by node $v_i$ to all of its two-hop neighbors, see $B$ above); (iv) the sets $\mathcal{T}_i[k]$ and $\mathcal{T}_i[k-1]$ are obtained via the one-hop transmissions by node $v_{i}$ at iterations $k$ and $k-1$. Table~\ref{TABinfojconcurrent} summarizes the information received at node $v_j$ at iterations $k-1$ and $k$ from neighbor $v_i$ and from a generic neighbor of node $v_i$, denoted by $v_l$ (note that $v_l$ is a two-hop neighbor of node $v_j$). Let us reemphasize that $\sigma_i[k]$ and $\sigma_i[k+1]$ could have been obtained via one-hop transmissions from node $v_i$ to node $v_j$; however, node $v_i$ also sends this information to all of its two-hop neighbors (following $B$ above). Thus, we treat it as two-hop information.

\begin{table}[htb]
\caption{Information received at $v_j$ for Concurrent Checking} \label{TABinfojconcurrent}
\begin{center}
\begin{tabular} {|c||c|c||c|}
\hline
 Iteration & One-hop $v_i$ & Two-hop $v_i$ & Two-hop $v_l$ \\
\hline \hline
$k-1$ & $\mathcal{T}_i[k-1]$ & $\sigma_i[k]$ & $\sigma_l[k]$ \\
          & $x_i[k]$ & & \\      				
\hline
$k$ & $\mathcal{T}_i[k]$ & $\sigma_i[k+1]$ & $\sigma_l[k+1]$ \\
       & $x_i[k+1]$ & & \\      				
\hline
\end{tabular}
\end{center}
\end{table}

\subsubsection{Concurrent Checking}

Node $v_j$ can directly assess whether node $v_{i}$ has correctly updated its values $x_{i}[k+1]$ and $\sigma_{i}[k+1]$ (by comparing the values it calculates against those broadcasted by node $v_i$ at the current iteration). More specifically, based on the information it has available (refer to Table~\ref{TABinfojconcurrent}), node $v_j$ performs the following parity checks at the end of iteration $k$:
\begin{eqnarray}
p_{ij}[k] & = & x_{i}[k+1] - \widehat{x}_i[k+1] \; , \label{EQpar1} \\ 
q_{ij}[k] & = & \sigma_{i}[k+1] - ( \sigma_i[k] + x_i[k]) \label{EQpar2} \; ,
\end{eqnarray}

\noindent 
where $\widehat{x}_i[k+1]$ is calculated by node $v_j$ as follows:
$$ 
\begin{array}{rcl}
\widehat{x}_i[k+1] & = & \left ( 1 - \frac{D_i[k]}{N} \right ) x_i[k] + \\
 & & ~ + \frac{1}{N} \sum_{v_l \in \mathcal{N}_i} (\mu_{il}[k+1] - \mu_{il}[k]) + \\ 
 & & ~ + (| \Delta \mathcal{U}_{i}[k] | - | \Delta \mathcal{T}_i[k] |) \frac{1}{N} \sigma_i[k] \; ,
\end{array}
$$
with $D_i[k] = | \mathcal{N}_i[k] |$, $\mathcal{N}_i[k] = \mathcal{N}_i \cap \mathcal{T}_i[k]$, $\Delta \mathcal{U}_i[k] = \mathcal{N}_i \cap (\mathcal{T}_i[k-1] \setminus \mathcal{T}_i[k])$, and $\Delta \mathcal{T}_i[k] = \mathcal{N}_i \cap (\mathcal{T}_i[k] \setminus \mathcal{T}_i[k-1])$. Note that $\mathcal{T}_i[k]$ and $\mathcal{T}_i[k-1]$ are reported to node $v_j$ directly by node $v_i$ and the $\mu_{il}$'s can be obtained from $\sigma_l$'s (just like node $v_i$ would obtain them):
$$
\mu_{il}[k+1]  = \left \{ 
\begin{array}{ll} 
\sigma_l[k+1] \;, & \forall v_l \in \mathcal{N}_i[k] , \\
0 \;, & \text{otherwise.}
\end{array} \right .
$$

Effectively, node $v_j$ checks the computation of node $v_i$ based on information reported by the neighbors of node $v_i$ and node $v_i$ itself. It is worth pointing out that the information from the neighbors of node $v_i$ is the same information as the one node $v_i$ is using. Thus, if node $v_j$ discovers that $p_{ij}[k] \neq 0$ and/or $q_{ij}[k] \neq 0$, then it can safely declare node $v_i$ as malicious (i.e., set $t_{ij}[t] = 0$ for all $t > k$). In fact, all other neighbors of node $v_i$ (not just node $v_j$) will also discover at iteration $k$ that node $v_i$ is misbehaving because their parity checks will evaluate to the same values as the ones in \eqref{EQpar1}--\eqref{EQpar2} (since they are based on identical information). This means that all neighbors of node $v_i$, which are all trustworthy under Assumption~2t, will declare node $v_i$ malicious at iteration $k$; this effectively isolates node $v_i$ from the remainder of the network.

\subsubsection{Parity Check Sufficiency}

The previous section described necessary conditions that need to be satisfied for node $v_i$ to be considered trustworthy by node $v_j$ (namely, the parity checks $p_{ij}[k]$ and $q_{ij}[k]$ need to be zero at each iteration $k$). Under some mild conditions, we argue in this section that checking these two parity checks is also sufficient.

One interesting case is the following: the above scheme allows node $v_i$ to try to manipulate the distributed computation by declaring one (or more) of its neighbors, say node $v_{l'}$, $v_{l'} \in \mathcal{N}_i$, as untrustworthy and then performing the correct computation (i.e., report the correct values under the assumption that node $v_{l'}$ is untrustworthy). Of course, this restricts how node $v_i$ can affect the outcome of the computation as it cannot arbitrarily change its values, but it still gives node $v_i$ a finite number of ways in which to alter the outcome of its computation (by including or excluding one of the $2^{D_i}$ different subsets\footnote{Node $v_i$ has $D_i$ neighbors and can declare one or more of them as untrustworthy; it has $2^{D_i}$ different ways of declaring its $D_i$ neighbors as trustworthy or untrustworthy.} of its neighbors via the set $\mathcal{T}_i[k]$ that it is reporting). Note that in such case, node $v_{l'}$, and any other (trustworthy) neighbor of node $v_i$ that might have been unfairly declared untrustworthy by node $v_i$, will immediately realize that node $v_i$ is untrustworthy (under Assumption~0t, node $v_{l'}$ becomes aware that node $v_i$ is declaring it untrustworthy since $v_{l'}$ also receives the set $\mathcal{T}_i[k]$). Thus, we assume that, if node $v_{l'}$, which is trustworthy is declared untrustworthy by node $v_i$, then node $v_{l'}$ will immediately declare node $v_i$ as untrustworthy. A potential problem, however, is the fact that the other neighbors of node $v_i$ (including node $v_j$) may not necessarily be in position to declare node $v_i$ as untrustworthy (as they may have no direct knowledge of the trustworthiness of node $v_{l'}$). In such case, the untrustworthy node $v_i$ is not removed but the combined actions of nodes $v_i$ and $v_{l'}$ effectively remove the link between nodes $v_i$ and $v_{l'}$. Note that the graph remains connected (recall that, under Assumption~2t, the graph $\mathcal{G}_{T}$ is assumed to be connected even if untrustworthy node $v_i$ and {\em all} of its links, not just the link between nodes $v_i$ and $v_{l'}$, are removed). Therefore, node $v_i$'s initial value will be included in the computation of the average, but node $v_i$ will not be able to affect the average computation in any other way. This is legitimate behavior under our assumptions (effectively node $v_i$ is not behaving maliciously, other than forcing a link in the graph to be removed). Note that node $v_i$'s initial value will be included in the average calculation (unless node $v_i$ misbehaves in a different manner).

In the above scenario, if we require that trustworthy nodes need to have $\mathcal{T}_i[k+1] \subseteq \mathcal{T}_i[k]$ for all $k$, node $v_i$ will easily be identified as untrustworthy if it tries to change the status of its neighbor $v_{l'}$ from trustworthy to untrustworthy and back  (e.g., in order to delay the convergence of the algorithm). Note, however, that Assumption~3 did not require this monotonicity: as long as, for $k > k_{\max}$, the sets $\mathcal{T}_j[k] = \mathcal{V}_T$ for all trustworthy nodes $v_j \in \mathcal{V}_T$,  Algorithm~\ref{algorithm_0} will allow nodes to converge to the trustworthy average.

The above distributed scheme allows nodes to eventually identify their untrustworthy neighbors and remove them from the computation. Since the graph $\mathcal{G}_{T}$ (after removing all untrustworthy nodes and the links associated with them) remains connected, the trustworthy nodes will eventually compute the trustworthy average. Untrustworthy nodes can slow down the computation (by removing links with neighboring nodes by declaring (unfairly) one or more of their trustworthy neighbors as untrustworthy) but cannot affect the computation in any other way (other than including their own initial value in the calculation, which is a legitimate action to take). Untrustworthy nodes can also sacrifice themselves to delay convergence to the average (e.g., by behaving correctly to remain in the computation and then, at some later time step, misbehave in order to cause a disturbance in the computation). %; however, they can only do that a finite number of times (at most as many times as their neighbors).

\begin{example}
We use this example to point out that if there are untrustworthy nodes that are neighbors (i.e., if Assumption~2t is violated), then the concurrent checking scheme described in this section is not guaranteed to capture the trustworthy nodes. Consider a bidirectional line graph of six nodes, $v_1$, $v_2$, $v_3$, ..., $v_6$, that are connected via edges between $v_1$ and $v_2$, between $v_2$ and $v_3$, ..., and between $v_5$ and $v_6$. If nodes $v_5$ and $v_6$ are untrustworthy, the induced graph involving trustworthy nodes is connected. However, since trustworthy nodes are only aware of their two-hop neighborhood, it is not possible for them to capture violations of the untrustworthy nodes: node $v_5$ can collude with node $v_6$ to pretend that any changes it incorporates in the computation are coming from neighbor $v_6$ (which cannot be checked by any trustworthy node). In other words, node $v_5$ can act correctly so that when node $v_4$ performs its checks does not identify any problems, at least based on what node $v_6$ is reporting; however, the only node that can check what node $v_6$ reports is node $v_5$, which is itself malicious.
\end{example}

\subsection{Two-Hop Information Received Infrequently}

In this section, we perform trust evaluations, similar to the ones developed in the previous section, but infrequently (more precisely, at random instants of time). The fact that trust evaluations are performed infrequently does not compromise the ability of nodes to capture misbehavior by their neighboring nodes (even when this misbehavior occurs in between checks) because we utilize the invariant established in Theorem~\ref{THEinvariant} about Algorithm~\ref{algorithm_0} in a way that guarantees that even if a node manipulates its values at an iteration during which it is not checked by any of its neighbors, its misbehavior (if significant\footnote{Misbehavior by node $v_i$ is insignificant if it does not change the average value the nodes converge to.}) will be captured when it is next checked by its neighbors. 

Under Assumptions~0t,~1, and~2t, let us consider that, at initialization, each node $v_\ell$ sends to its (immediate) neighbors the values $x_\ell[0]=x_\ell$, $\sigma_\ell[0]=0$, and $\mathcal{T}_\ell[-1] = \mathcal{V}$. Subsequently, at the end of each iteration $k$ of Algorithm~1 ($k=0, 1, 2, ...$), each node $v_\ell$ sends to its (immediate) neighbors its values $x_\ell[k+1]$, $\sigma_\ell[k+1]$, and $\mathcal{T}_\ell[k]$. At each iteration $k$, node $v_\ell$ also saves its previous $\sigma_\ell[k]$ value. 

The above communication strategy is followed by all nodes, including each neighbor $v_i$ of node $v_j$ ($v_i \in \mathcal{N}_j$). We next describe the checking performed from the point of view of node $v_j$ (each node does something similar). At random points in time that node $v_j$ selects, node $v_j$ initiates a check of all of the nodes in its neighborhood. More specifically, at the end of a randomly selected iteration, denoted here by $K$, node $v_j$ requests and receives additional information from its two-hop neighbors in $\mathcal{N}_j^{(2)} := \mathcal{N}_j \cup (\cup_{v_i \in \mathcal{N}_j} \mathcal{N}_i)$. At this point, each two-hop neighbor $v_l \in \mathcal{N}_j^{(2)}$ sends to node $v_j$ the following information: 
\begin{enumerate}
\item its current running sum $\sigma_l[K+1]$, and
\item its previous running sum $\sigma_l[K]$ (which is information that each node stores for one time step).
\end{enumerate}

Table~\ref{TABinfojnonconcurrent} summarizes the information received at node $v_j$ at iterations $K-1$ and $K$ from neighbor $v_i$ and from a generic neighbor of node $v_i$, denoted by $v_l$ (note that $v_l$ is a two-hop neighbor of node $v_j$). We point out that two-hop information is only sent at iteration $K$ (when a check is initiated by node $v_j$) and involves the running sums $\sigma_l[K]$ (one iteration step ago) and $\sigma_l[K+1]$ (current iteration) for each $v_l \in \mathcal{N}_j^{(2)}$. Note that we denote these running sums by $\widetilde{\sigma}_l[K]$ and $\widetilde{\sigma}_l[K+1]$ because they could be different from the actual runnings sum $\sigma_l[K]$ and $\sigma_l[K+1]$ (the latter were sent by node $v_l$ to its one-hop neighbors at iterations $K-1$ and $K$ respectively). In particular, if node $v_l$ is malicious, we could have $\widetilde{\sigma}_l[K] \neq \sigma_l[K]$ and/or $\widetilde{\sigma}_l[K+1] \neq \sigma_l[K+1]$.

\begin{table}[htb]
\caption{Information received at $v_j$ for Non-Concurrent Checking} \label{TABinfojnonconcurrent}

\begin{center}
\begin{tabular} {|c||c||c|}
\hline
 Iteration & One-hop $v_i$ & Two-hop $v_l$ \\
\hline \hline
$K-1$ & $\mathcal{T}_i[K-1]$, $\sigma_i[K]$, & \\ 
           & $x_i[K]$ & \\      				
\hline
$K$ & $\mathcal{T}_i[K]$, $\sigma_i[K+1]$, & $\widetilde{\sigma}_l[K]$, \\
        & $x_i[K+1]$ & $\widetilde{\sigma}_l[K+1]$ \\
\hline
\end{tabular}
\end{center}

\end{table}

\subsubsection{Infrequent Checking}

At the end of iteration $K$, node $v_j$ can check the computations performed by each neighbor $v_i$ ($v_i \in \mathcal{N}_j$), as it has access to all information that is needed from node $v_i$ to execute an iteration step. More specifically, node $v_j$ knows: (i) the running sums $\{ \widetilde{\sigma}_{l}[K] \; | \; v_{l} \in \mathcal{N}_i \}$, obtained via the two-hop transmissions by each node $v_l$ in $\mathcal{N}_i$ at iteration $K$ ($v_l \in \mathcal{N}_i$, thus $v_l \in \mathcal{N}_j^{(2)}$); (ii) the value $x_{i}[K]$, obtained via the one-hop transmission by node $v_i$ at iteration $K-1$; (iii) the value $\sigma_{i}[K]$, obtained via the one-hop transmission by node $v_i$ at iteration $K-1$; (iv) the sets $\mathcal{T}_i[K-1]$ and $\mathcal{T}_i[K]$, obtained via the one-hop transmissions of node $v_{i}$ at iterations $K-1$ and $K$, respectively.

Therefore, node $v_j$ can directly assess whether node $v_{i}$ has correctly updated its values $x_{i}[K+1]$ and $\sigma_{i}[K+1]$ (by comparing the values it calculates against those broadcasted by node $v_i$ at the current iteration). More specifically, based on the information it has available (refer to Table~\ref{TABinfojnonconcurrent}), node $v_j$ performs the following parity checks:
\begin{eqnarray}
p_{ij}[K] & = & x_{i}[K+1] - \widehat{x}_i[K+1] \; , \label{EQpar1n} \\ 
q_{ij}[K] & = & \sigma_{i}[K+1] - ( \sigma_i[K] + x_i[K]) \label{EQpar2n} \; ,
\end{eqnarray}
where $\widehat{x}_i[K+1]$ is calculated by node $v_j$ as 
$$ 
\begin{array}{rcl}
\widehat{x}_i[K+1] & = & \left ( 1 - \frac{D_i[K]}{N} \right ) x_i[K] \\ 
 & & ~ + \frac{1}{N} \sum_{v_l \in \mathcal{N}_i} (\widetilde{\mu}_{il}[K+1] - \widetilde{\mu}_{il}[K]) + \\ 
 & & ~ + (| \Delta \mathcal{U}_{i}[K] | - | \Delta \mathcal{T}_i[K] |) \frac{1}{N} \sigma_i[K] \; ,
\end{array}
$$
with $D_i[k] = | \mathcal{N}_i[k] |$, $\mathcal{N}_i[k] = \mathcal{N}_i \cap \mathcal{T}_i[k]$, for $k=K-1$ and $k=K$, and $\Delta \mathcal{U}_i[K] = \mathcal{N}_i \cap (\mathcal{T}_i[K-1] \setminus \mathcal{T}_i[K])$, and $\Delta \mathcal{T}_i[K] = \mathcal{N}_i \cap (\mathcal{T}_i[K] \setminus \mathcal{T}_i[K-1])$. Note that $\mathcal{T}_i[K]$ and $\mathcal{T}_i[K-1]$ are reported to node $v_j$ directly by node $v_i$, but the $\widetilde{\mu}_{il}$'s are obtained from $\widetilde{\sigma}_l$'s as follows:
$$
\widetilde{\mu}_{il}[k+1]  = \left \{ 
\begin{array}{ll} 
\widetilde{\sigma}_l[k+1] \;, & \forall v_l \in \mathcal{N}_i[k] , \\
0 \;, & \text{otherwise,}
\end{array} \right .
$$
for $k=K-1$ and $k=K$.

In addition, node $v_j$ checks that the invariant in \eqref{EQinvariant} holds for node $v_i$ by performing the following check:
\begin{eqnarray}
r_{ij}[K] & = & x_{i}[K+1] - x_i + \frac{D_i[K]}{N} \sigma_i[K+1] - \nonumber \\ 
 & & ~ - \sum_{v_l \in \mathcal{N}_i[K]} \widetilde{\sigma}_l[K+1] \; , \label{EQpar3n}
\end{eqnarray}
where $x_i$ became available to node $v_j$ at the initialization of the algorithm.

Finally, node $v_j$ checks whether the running sums transmitted by node $v_i$ to node $v_i$'s neighbors and to node $v_i$'s two-hop neighbors are in agreement (i.e., $\widetilde{\sigma}_i[K] \neq \sigma_i[K]$ and/or $\widetilde{\sigma}_i[K+1] \neq \sigma_i[K+1]$); this can be done via the following parity check: 
\begin{equation}
s_{ij}[K] = | \widetilde{\sigma}_i[K] - \sigma_i[K] | + | \widetilde{\sigma}_i[K+1] - \sigma_i[K+1] | \; . \label{EQpar4n}
\end{equation}

We assume that, when node $v_j$ initiates a check at iteration $K$ for node $v_i$ (and all of node $v_j$'s neighbors), all other neighbors of node $v_i$ (not just node $v_j$) will also perform at iteration $K$ the exact same checks in \eqref{EQpar1n}--\eqref{EQpar4n} as node $v_j$ (this is possible because they receive information similar to the one received by node $v_j$). For example, node $v_l$ will also be able to check node $v_i$ though it will not necessarily be in position to check its other neighbors. This means that all neighbors of node $v_i$ will reach a decision regarding the trustworthiness of node $v_i$; as we explain next, however, the decisions of the neighbors of node $v_i$ regarding the status of node $v_i$ need not coincide (though they will eventually coincide).

\subsubsection{Parity Check Analysis}

Node $v_j$ checks the computation of node $v_i$ based on information reported by the neighbors of node $v_i$ and node $v_i$ itself. If node $v_j$ discovers that $q_{ij}[K] \neq 0$ or $s_{ij}[K] \neq 0$, then it can safely declare node $v_i$ to be malicious, i.e., set $t_{ij}[t] = 0$ for all $t > K$, because $q_{ij}[K]$ and $s_{ij}[K]$ are computed purely on information provided by node $v_i$ and one of them being nonzero is an indication that node $v_i$ has provided inconsistent information). However, if node $v_j$ discovers that $p_{ij}[K] \neq 0$ and/or $r_{ij}[K] \neq 0$, then it cannot safely declare node $v_i$ as malicious. Then, there are two cases to be considered: 
\\ Case 1: node $v_i$ is malicious; and/or 
\\ Case 2: one or more neighbors of node $v_i$, say node $v_l$, has reported an incorrect running sum (i.e., $\widetilde{\sigma}_l[K] \neq \sigma_l[K]$ and/or $\widetilde{\sigma}_l[K+1] \neq \sigma_l[K+1]$).

We first discuss Case 2, which is more straightforward. If node $v_{l}$ is sending to its two-hop neighbors a different running sum than the one that it sent to its one-hop neighbors (i.e., $\widetilde{\sigma}_l[K] \neq \sigma_l[K]$ and/or $\widetilde{\sigma}_l[K+1] \neq \sigma_l[K+1]$), this will prompt all (trustworthy) neighbors of node $v_{l}$ (including node $v_i$, if trustworthy) to declare node $v_{l}$ as untrustworthy. This effectively removes node $v_{l}$ from the computation (under Assumption~2t, all neighbors of node $v_{l}$ are trustworthy and will set their trust assessment about node $v_{l}$ to zero). This includes node $v_i$ itself because $v_i$ will also realize that its neighboring node $v_{l}$ is acting maliciously (reporting mismatched running sums, which will manifest itself as a violation of the fourth parity check, i.e., $s_{li}[K] \neq 0$). This means that, for subsequent iterations $k>K$, $\mathcal{T}_i[k]$ will {\em not} include node $v_{l}$.

Note that, at iteration $K$, node $v_j$ cannot be certain about the status of node $v_i$ (because it cannot discriminate between Case~1 and Case~2). For this reason, node $v_j$ (as well as all other trustworthy nodes that are neighbors of node $v_i$) will consider node $v_i$ to be ``possibly untrustworthy," which means that it might be (permanently) declared untrustworthy when these nodes perform their next checks. In particular, after the check at iteration $K$, node $v_j$ expects node $v_i$ (if trustworthy) to remove at least one node from its computation, i.e., to report a set $\mathcal{T}_i[K+1]$ that is strictly contained in $\mathcal{T}_i[K]$. If that does not happen at the next iteration, then node $v_j$ can safely declare node $v_i$ to be untrustworthy. Note that node $v_j$ does not necessarily know which neighbor of node $v_i$ might be untrustworthy (it is simply aware that there is a disagreement in terms of what node $v_i$ and its neighbors are reporting); however, node $v_j$ is expecting node $v_i$, if trustworthy, to remove at least one of its neighboring nodes from the set $\mathcal{T}_i[K+1]$ (otherwise, Case~1 holds and node $v_j$ declares node $v_i$ untrustworthy at iteration $K+1$).

Note that node $v_j$ can also easily check whether or not $\mathcal{T}_i[k+1] \subseteq \mathcal{T}_i[k]$, which should hold for all $k$ (as it receives this information at each iteration). In particular, if node $v_i$ was considered to be ``possibly untrustworthy" (due to the check at iteration $K$ when node $v_j$ last initiated two-hop information exchange), then strict inequality $\mathcal{T}_i[K+1] \subset \mathcal{T}_i[K]$ should hold. It is possible for a node $v_i$ that was deemed ``possibly untrustworthy" by node $v_j$ to present a $\mathcal{T}_{i}[K+1]$ that is strictly contained in $\mathcal{T}_{i}[K]$ (and thus become trustworthy), but node $v_j$ to identify another inconsistency at a subsequent check, say at iteration $K'$. In such case, node $v_j$ again considers node $v_i$ to be ``possibly untrustworthy" because Case~1 is not certain (it could be that Case~2 holds for a different neighbor $v_{l'}$ of node $v_{i}$). Of course, the procedure continues because at iteration $K'+1$ node $v_j$ expects $\mathcal{T}_{i}[K'+1]$ to be reduced even further (otherwise, node $v_j$ can safely declare node $v_i$ as untrustworthy). Note that this can only happen a finite number of times (at most as many as the number of neighbors of node $v_i$).

\subsubsection{Parity Check Sufficiency}

The reason one also needs to check the invariant in \eqref{EQinvariant} (via $r_{ij}$) is because an untrustworthy node $v_i$ has the opportunity to change its $x$ and $\sigma$ values arbitrarily during the time steps when it is not being checked. Of course, if it attempts such a change, $v_i$ is risking getting caught by node $v_j$ (or by any other of its trustworthy neighbors) because the latter may randomly decide to request two-hop information. Nevertheless, if node $v_i$ takes the risk at one iteration and does not get caught, node $v_j$ will not be able to determine this misbehavior if it only checks the computational updates at later iterations (because those updates might be correctly performed by node $v_i$ and the damage has already been done). Checking the invariant ensures that node $v_i$ will get caught at a later iteration (unless its invariant is the proper one, which guarantees that node $v_i$ is contributing the correct initial value).

In the above scenario, it is possible for node $v_i$ to misbehave multiple times during checks and not get caught if it somehow manages to present the correct invariant when the checks are taking place. In fact, if the invariant for node $v_i$ is correct,  Theorem~\ref{THEmain} effectively implies that node $v_i$ is behaving correctly (we know that nodes will converge to the correct average). Therefore, node $v_i$ can indeed manipulate its updates in-between checks without getting caught if it somehow manages to report $x$ and $\sigma$ values that satisfy the update check and the invariant check (for instance, it might attempt to do that in order to delay convergence); however, such manipulations cannot happen indefinitely because node $v_i$ will eventually get caught, since, due to the randomness of the checks of node $v_j$, node $v_i$ will be identified as untrustworthy with probability one.

\begin{remark}
As in the previous section, the above scheme allows node $v_i$ to try to manipulate the average by declaring one (or more) of its neighbors, say node $v_{l'}$, as untrustworthy and then performing the correct computation (i.e., report the correct values under the assumption that node $v_{l'}$ is untrustworthy). This will effectively remove the link $(v_i, v_{l'})$ and can only happen a finite number of times. Under such strategy, node $v_i$ manages to incorporate its value in the calculation, but cannot affect the average in any other way. Furthermore, if node $v_i$ its trustworthy, the set $\mathcal{T}_i[k+1]$ is a subset of $\mathcal{T}_i[k]$ for all $k$, and this is something that can be verified by its neighbors at each iteration. Note that $\mathcal{T}_i[k+1] \subseteq \mathcal{T}_i[k]$ for all $k$ is also a feature of the scheme proposed in \cite{yuan2019resilient, yuan2021resilient, yuan2021secure}. However, in our case, node $v_j$ may assign the status ``possibly untrustworthy" to its neighbor $v_i$, while waiting for node $v_i$ to determine any untrustworthy neighbors it might have.
\end{remark}

\section{Conclusions and Future Work}
\label{SECconclusions}

In this paper, we have considered the problem of trustworthy distributed average consensus in multi-agent systems, in the presence of malicious nodes that may try to influence the outcome of the computation via their updates that might be chosen in a collusive manner. The proposed algorithm allows the nodes to asymptotically converge to the average of the initial values of the trustworthy nodes, assuming that  (i) the underlying bidirectional communication topology that describes the information exchange among the non-malicious nodes is connected, and (ii) the non-malicious nodes eventually receive correct information about the trustworthiness of other nodes. The proposed algorithm allows the nodes to continuously adjust their values and updating strategy as they receive new information about the trustworthiness of other nodes; assuming that eventually this information correctly represents the trustworthiness of the various nodes in the distributed system, the non-malicious nodes asymptotically converge to the average of the initial values of the trustworthy nodes. When the nodes are capable of performing (perhaps periodically and at a higher cost) two-hop communication transmissions, we have also proposed strategy for the nodes to perform checks and obtain the trust assessments, in a way that guarantees that all malicious nodes are eventually identified by the trustworthy nodes, effectively ensuring convergence to the trustworthy average.

In our future work we plan to further research into mechanisms for distributively obtaining the trust assessments that are required by Algorithm~1 (e.g., in directed communication topologies) and/or relaxing the topological requirements we imposed (e.g., consider cases where neighboring nodes may be malicious). We also plan to consider dynamic versions of the above problem where nodes can update the initial measurement, or enter and leave the distributed system.

\bibliographystyle{IEEEtran}

\bibliography{references,bibliografia_consensusV4}

\begin{IEEEbiography}[{\includegraphics[width=1in,height=1.25in,clip,keepaspectratio]{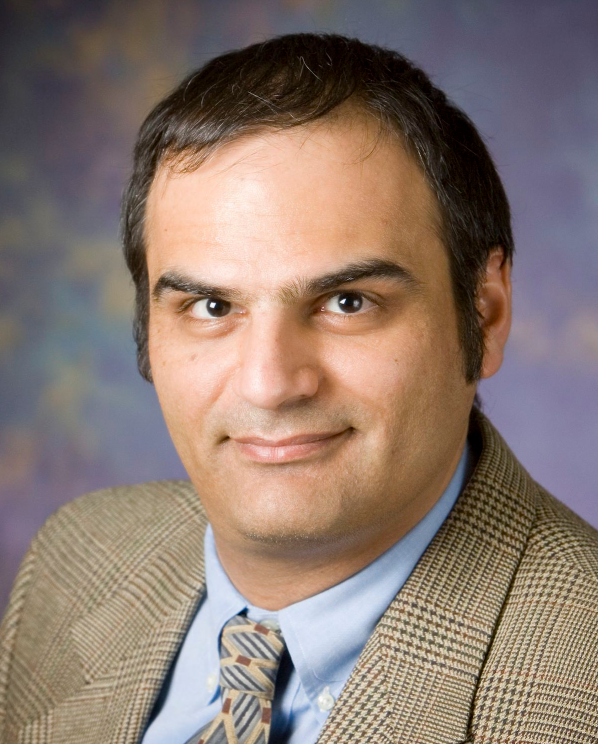}}]{Christoforos N. Hadjicostis} 
(M'99, SM'05, F'20) received the S.B. degrees in electrical engineering, computer science and engineering, and in mathematics, the M.Eng. degree in electrical engineering and computer science in 1995, and the Ph.D. degree in electrical engineering and computer science in 1999, all from the Massachusetts Institute of Technology, Cambridge. In 1999, he joined the Faculty at the University of Illinois at Urbana-Champaign, where he served as Assistant and then Associate Professor with the Department of Electrical and Computer Engineering, the Coordinated Science Laboratory, and the Information Trust Institute. Since 2007, he has been with the Department of Electrical and Computer Engineering, University of Cyprus, where he is currently Professor and Interim Director of the Daedalus Research Center. His research focuses on fault diagnosis and tolerance in distributed dynamic systems, error control coding, monitoring, diagnosis and control of large-scale discrete-event systems, and applications to network security, anomaly detection, and energy distribution systems. Dr. Hadjicostis serves as Editor in Chief of the Journal of Discrete Event Dynamic Systems and as Senior Editor of IEEE Transactions on Automatic Control. In the past, he served as Associate Editor of Automatica, IEEE Transactions on Automation Science and Engineering, IEEE Transactions on Control Systems Technology, IEEE Transactions on Circuits and Systems~I, and the Journal of Nonlinear Analysis of Hybrid Systems. 
\end{IEEEbiography}

\begin{IEEEbiography}[{\includegraphics[width=1.05in,height=1.25in,clip,keepaspectratio]{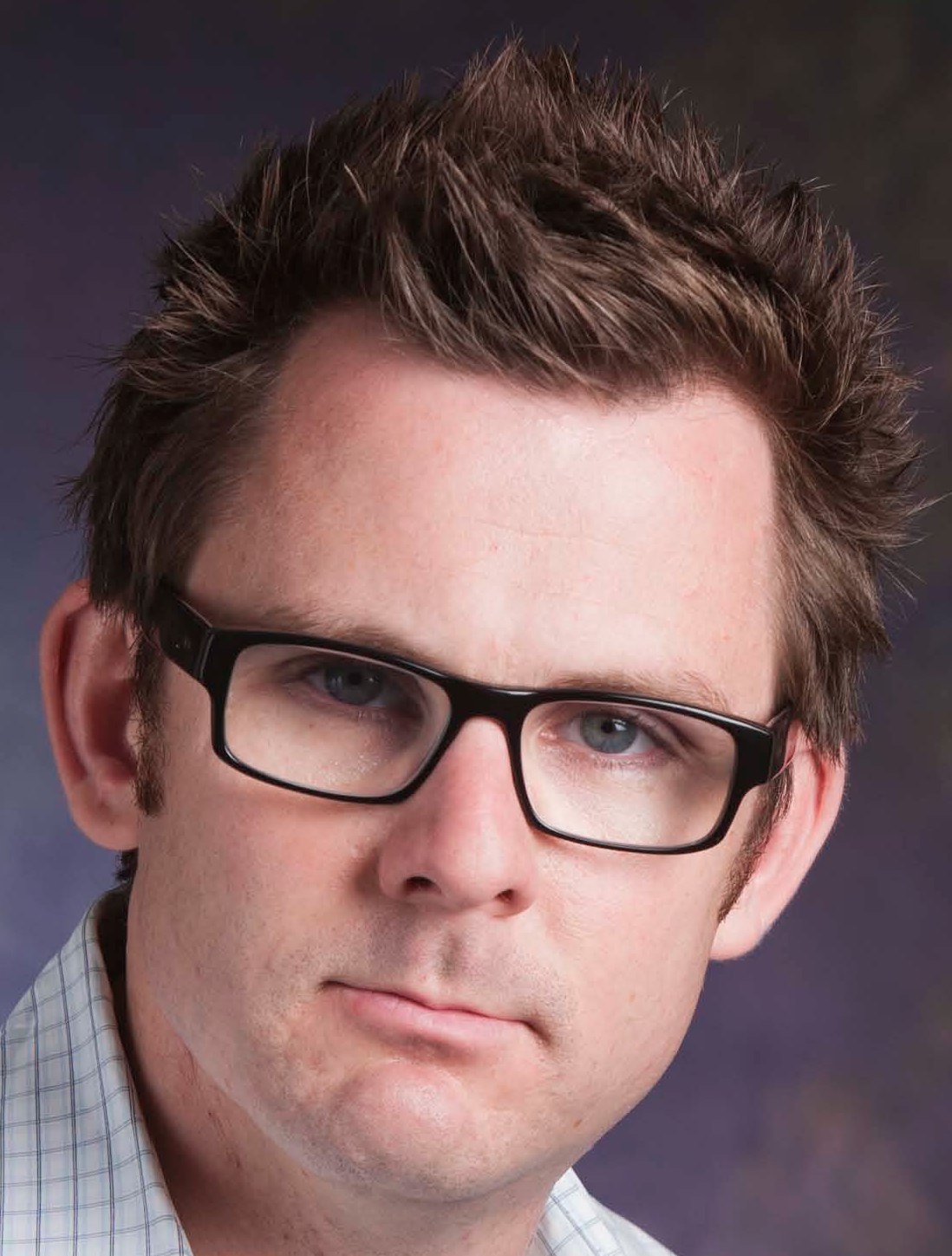}}]{Alejandro D. Dom\'{i}nguez-Garc\'{i}a} (S'02, M'07, SM'20, F'23) received the degree of Electrical Engineer from the University of Oviedo (Spain) in 2001 and the Ph.D. degree in electrical engineering and computer science from the Massachusetts Institute of Technology, Cambridge, MA, in 2007. 

He is Professor with the Department of Electrical and Computer Engineering (ECE), and Research Professor with the Coordinated Science Laboratory and the Information Trust Institute, all at the University of Illinois at Urbana-Champaign.  He is affiliated with the ECE Power and Energy Systems area, and  has been a Grainger Associate since August 2011.   

His research interests are in the areas of system reliability theory and control, and their applications to electric power systems, power electronics, and embedded electronic systems for safety-critical/fault-tolerant aircraft, aerospace, and automotive applications.

Dr. Dom\'{i}nguez-Garc\'{i}a received the NSF CAREER Award in 2010, and the Young Engineer Award from the IEEE Power and Energy Society in 2012. In 2014, he was invited by the National Academy of Engineering to attend the US Frontiers of Engineering Symposium, and was selected by the University of Illinois at Urbana-Champaign Provost to receive a Distinguished Promotion Award.  In 2015, he received the U of I College of Engineering Dean's Award for Excellence in Research. 

He is currently an editor for the IEEE Transactions on Control of Network Systems; he also served as an editor of the IEEE Transactions on Power Systems and IEEE Power Engineering Letters from 2011 to 2017.
\end{IEEEbiography}

\end{document}